\documentclass[11pt]{article}

\usepackage{amsmath,amsthm,amscd,amsfonts,amssymb}
\usepackage{mathrsfs}
\usepackage{hyperref}
\usepackage[all]{xy}
\usepackage[usenames,dvipsnames,svgnames,table]{xcolor}

\definecolor{darkred}{RGB}{175,0,0}

 1 
\newtheorem{theorem}{Theorem}
\newtheorem{proposition}{Proposition}
\newtheorem{corollary}{Corollary}
\newtheorem{lemma}{Lemma}
\newtheorem{definition}{Definition}

\theoremstyle{definition}
\newtheorem*{remarkth}{Remark}
\newenvironment{remark}{\begin{remarkth}}{\hfill$\lozenge$\end{remarkth}}

\def\derpar#1#2{\displaystyle\frac{\partial{#1}}{\partial{#2}}}
\def\derpars#1#2#3{\displaystyle\frac{\partial^2{#1}}{\partial{#2}\partial{#3}}}
\def\restric#1#2{\left.#1\right|_{#2}}

\newcommand{\X}{\mathcal{X}}
\newcommand{\D}{\mathcal{D}}

\newcommand{\vf}{\mathfrak{X}}
\newcommand{\df}{\mathit{\Omega}}
\newcommand{\Lag}{\mathcal{L}}
\newcommand{\Leg}{\mathcal{FL}}
\renewcommand{\d}{\textnormal{d}}

\newcommand{\R}{\mathbb{R}}
\newcommand{\Id}{\textnormal{Id}}
\renewcommand{\Im}{\operatorname{Im}}
\renewcommand{\graph}{\operatorname{graph}}
\def\Tan{{\rm T}}
\def\Lie{\mathop{\rm L}\nolimits}
\def\inn{\mathop{i}\nolimits}
\def\Cinfty{{\rm C}^\infty}
\def\tabaddress#1{{\small\it\begin{tabular}[t]{c}#1
\\[1.2ex]\end{tabular}}}
\def\qed{\ifvmode\removelastskip\fi
{\unskip\nobreak\hfil\penalty50\hbox{}\nobreak\hfil \hbox{\vrule
height1.2ex width1.2ex}\parfillskip=0pt \finalhyphendemerits=0
\par\smallskip}}

\title{HAMILTON-JACOBI THEORY IN MULTISYMPLECTIC CLASSICAL FIELD THEORIES}

\author{
{\sc  Manuel de Le\'{o}n\thanks{\textbf{e}-{\it mail}: mdeleon@icmat.es} }\\
\vspace{5mm}
   \tabaddress{Instituto de Ciencias Matem\'{a}ticas (CSIC-UAM-UC3M-UCM). \\
   C/ Nicol\'{a}s Cabrera 15. Campus Cantoblanco UAM. 28049 Madrid. Spain} \\
{\sc  Pedro Daniel Prieto-Mart\'{\i}nez\thanks{\textbf{e}-{\it mail}: peredaniel@ma4.upc.edu} }\\
{\sc Narciso Rom\'an-Roy\thanks{\textbf{e}-{\it mail}: nrr@ma4.upc.edu}}  \\
\vspace{5mm}
\tabaddress{Departamento de Matem\'atica Aplicada IV. Edificio C-3, Campus Norte UPC\\
   C/ Jordi Girona 1. 08034 Barcelona. Spain}\\
{\sc  Silvia Vilari\~{n}o\thanks{\textbf{e}-{\it mail}: silviavf@unizar.es} }\\
\vspace{5mm}
   \tabaddress{Centro Universitario de la Defensa de Zaragoza \& I.U.M.A. \\
   Academia General Militar. Carretera de Huesca s/n. 50090 Zaragoza, Spain.} \\   
}
\date{\today}

\begin{document}

\maketitle

\pagestyle{myheadings}

\thispagestyle{empty}

\begin{abstract}
The geometric framework for the Hamilton-Jacobi theory developed in
\cite{art:Carinena_Gracia_Marmo_Martinez_Munoz_Roman06,HJteam-2015,LMM-09} is extended for
multisymplectic first-order classical field theories. The Hamilton-Jacobi problem is stated for
the Lagrangian and the Hamiltonian formalisms of these theories as a particular case of a more
general problem, and the classical Hamilton-Jacobi equation for field theories is recovered from
this geometrical setting. Particular and complete solutions to these problems are defined and
characterized  in several equivalent ways in both formalisms, and the equivalence between them is
proved. The use of distributions in jet bundles that represent the solutions to the field equations
is the fundamental tool in this formulation. Some examples are analyzed and, in particular, the
Hamilton-Jacobi equation for non-autonomous mechanical systems is obtained as a special case of our results.
\end{abstract}


\vbox{\raggedleft AMS s.\,c.\,(2010): 53C15, 53C80, 70S05, 35F21, 70H20}\null 
\markright{\rm M. de Le\'on \textit{et al.}: \textsl{Hamilton-Jacobi theory in multisymplectic classical field theories}}

\clearpage

\tableofcontents

\section{Introduction}
\label{intro}

The Hamilton-Jacobi theory, as it is classically presented in the textbooks and works on analytical
mechanics, is a way to integrate Hamilton equations (that is, a system of first-order ordinary
differential equations), which consists in giving an appropriate canonical transformation leading
the system to equilibrium \cite{Arn,JS,Sa71}. This transformation is constructed from its generating
function which, in this method, is obtained as the solution to a partial differential equation: the
so-called Hamilton-Jacobi equation. This method is based on a famous contribution from Hamilton on
geometric optics, where he showed that the propagation of wavefronts is characterized by a function
(the characteristic function) which is the solution to a first-order partial differential equation
called \textsl{eikonal equation}, which is related to the Hamilton-Jacobi equation. Thus, from a
purely mathematical perspective, the Hamilton-Jacobi theory is a very important example showing the
deep connection between systems of first-order ordinary differential equations and first-order partial
differential equations \cite{book:Rund66}. The Hamilton-Jacobi equation appears also when short-wave
approximations for the solutions of wave-type (hyperbolic) equations are considered. In this way,
from a physical point of view, being a classical equation, it is also very close to the Schr\"odinger
equation of quantum mechanics, since from a complete solution to the Hamilton-Jacobi equation, we are
able to reconstruct an approximate solution to the Schr\"odinger equation \cite{EMS04,MMMcq} and thus
it allows to establish an approach within classical theory of the notions of wave function and state
in quantum theory.

For all these reasons, Hamilton-Jacobi theory is a matter of continuous interest, and it was studied
and generalized also in other classical ambients; in particular, for constrained systems arising from
singular Lagrangians (gauge theories)  \cite{Gomis1} or also for higher-order dynamics \cite{art:Constantelos84}.

Furthermore, in the last decades, great efforts have been done in understanding physical systems from
a geometric perspective. Concerning to geometric mechanics, the intrinsic formulation of Hamilton-Jacobi
equation is also clear and can be found in \cite{AM,LM,MMM}. In addition,
in \cite{art:Carinena_Gracia_Marmo_Martinez_Munoz_Roman06} a generic geometric framework for the
Hamilton-Jacobi theory was formulated both in the Lagrangian and the Hamiltonian formalisms, for
autonomous and non-autonomous mechanics, recovering the usual Hamilton-Jacobi equation as a special
case in this generalized framework. In particular, it is shown that the existence of constants of motion
helps to solve the Hamilton-Jacobi problem, which can be regarded as a way to describe the dynamics
on the phase space of the system in terms of a family of vector fields on a submanifold of it. The
basic ideas of this generalization of the Hamilton-Jacobi formalism are similar to those outlined in
\cite{KV-1993}.

These geometric frameworks have been used by other authors to develop the Hamilton-Jacobi theory in
many different situations in mechanics. For instance, the statement and applications of the
Hamilton-Jacobi method for non-holonomic and holonomic mechanical systems is done in
\cite{BFS-14,HJnh,leones1,leones2,blo,OFB-11}, the geometric treatment of the theory for dynamical
systems described by singular Lagrangians is analyzed in \cite{LOS-12,LMV-12,LMV-13}, the application to
control theory is given in \cite{BLMMM-12,Wang1,Wang2}, and the generalization for higher-order
dynamical systems is established in \cite{art:Colombo_DeLeon_Prieto_Roman14_JPA,CLPR2}. Moreover,
the Hamilton-Jacobi theory has been extended for mechanical systems which are described using more
general geometrical frameworks, such as Lie algebroids \cite{BMMP-10,LS-12}, almost-Poisson manifolds
\cite{LMV-14}, and fiber bundles in general \cite{HJteam-2015}, and the relationship between the
Hamilton-Jacobi equation and some other geometric structures in mechanics are analyzed in
\cite{BLM-12,HJteam-K}. Finally, the geometric discretization of the Hamilton-Jacobi equation is
also considered in \cite{BDM-12,OBL-11}.

The extension of the Hamilton-Jacobi theory for first-order classical field theories has been
developed for different covariant formulations ($k$-symplectic and $k$-cosymplectic) in the
Hamiltonian formalism \cite{LMMSV-12,DeLeon_Vilarino} and also for the non-covariant Hamiltonian
formulation (Cauchy data space) \cite{CLMV-14}. A first quick approach to state the Hamilton-Jacobi
equation for the most general framework (the multisymplectic one) was outlined in \cite{LMM-09},
also in the Hamiltonian formalism. Furthermore, using a different approach involving connections,
the theory has been generalized to higher-order field theories \cite{Vi-10} and also for partial
differential equations in general \cite{Vi-11}.

The aim of this paper is to complete these previous developments; that is, to use the guidelines
stated in the aforementioned references on the Hamilton-Jacobi theory in geometric mechanics in
order to give a complete description of this theory for the multisymplectic formalism of first-order
classical field theories, both in the Lagrangian and the Hamiltonian formalisms, and showing the
equivalence between them. Our standpoint is \cite{LMM-09} and, in particular, some of our results
are a development of the ideas pointed out in this reference. As a fundamental difference with these
previous works, we consider the sections which are solutions of the field equations as integral
sections of integrable distributions in the corresponding phase spaces (jet bundles and bundles of
forms) where the equations are defined, and we represent these distributions by means of (classes of)
multivector fields in general \cite{art:Echeverria_Munoz_Roman98,EMR-99b}. This allows us to adapt the
geometric models for the Hamilton-Jacobi problem in mechanics given in
\cite{art:Carinena_Gracia_Marmo_Martinez_Munoz_Roman06,HJteam-2015} to the present case.

The paper is organized as follows: Section \ref{sec:Background} is a short review on multisymplectic
geometry, jet bundles, and multivector fields and their relation with integrable distributions, which
is given in order to do the paper self-contained. The main results are presented in Sections
\ref{sec:LagrangianFormalism} and \ref{sec:HamiltonianFormalism}, where first the generalized
Hamilton-Jacobi problem, and later the standard Hamilton-Jacobi problem are stated in the Lagrangian
and Hamiltonian formalisms, and the equivalence between both formalisms is analyzed. In these Sections,
the particular and complete solutions of the Hamilton-Jacobi equations are introduced and interpreted
geometrically. In Section \ref{sec:Examples}, some examples are studied; in particular, non-autonomous
dynamical systems as the particular case of a field theory with $1$-dimensional base manifold,
quadratic Lagrangian densities, and the problem of minimal surfaces in dimension three. Finally, the
conclusions and further research are presented in Section \ref{sec:Conclusions}, where the comparison
and differences between our model and the aforementioned previous works are also discussed.

All the manifolds are real, second countable and $\Cinfty$. The maps and the structures are assumed
to be $\Cinfty$. Sum over crossed repeated indices is understood.

\section{Geometrical background}
\label{sec:Background}

\subsection{Multisymplectic geometry}

In this section we give a short review on multisymplectic geometry and some particular
submanifolds of a multisymplectic manifold
(see \cite{art:Cantrijn_Ibort_DeLeon96,art:Cantrijn_Ibort_DeLeon99,art:Echeverria_Ibort_Munoz_Roman12}
for details).

Let $M$ be an $m$-dimensional smooth manifold. A \textsl{multisymplectic $k$-form} in $M$ is a closed
$k$-form $\omega \in \df^{k}(M)$ which, in addition, is $1$-nondegenerate, that is, for every $p \in M$,
$\inn(X_p)\omega_p = 0$ if, and only if, $X_p = 0$, where $X_p \in \Tan_pM$. If $\omega$ is closed and
$1$-degenerate, it is called a \textsl{premultisymplectic $k$-form}. A manifold endowed with a
(pre)multisymplectic form is called a \textsl{(pre)multisymplectic manifold of order $k$}.

Observe that a necessary condition for a $k$-form to be $1$-nondegenerate is $1 < k \leqslant \dim M$.

Given a symplectic manifold, we have a natural definition of ``orthogonality'' in terms of the
symplectic form. This definition can be generalized to multisymplectic manifolds, bearing in mind
that there are several levels of orthogonality to be considered.

\begin{definition}
Let $(M,\omega)$ be a multisymplectic manifold of order $k$, and $F \subseteq \Tan M$ a vector
subbundle. The \textnormal{$l$th orthogonal complement} of $F$, with $1 \leqslant l < k$ is the
subbundle $F^{\bot,l} \subseteq \Tan M$ defined as
\begin{equation*}
F^{\bot,l} = \left\{ (p,u_p) \in \Tan M \mid \omega_p(u_p,v_1,\ldots,v_l) = 0 \text{ for every }
(p,v_i) \in F \right\} \, .
\end{equation*}
\end{definition}

\begin{definition}
A subbundle $F \subset \Tan M$ is called \textnormal{$l$-isotropic} if $F \subseteq F^{\bot,l}$,
\textnormal{$l$-coisotropic} if $F^{\bot,l} \subseteq F$, and \textnormal{$l$-Lagrangian} if $F = F^{\bot,l}$,
for $1 \leqslant l < k$.
\end{definition}

Bearing in mind this last Definition, one can generalize the concepts of $l$-isotropic, $l$-coisotropic
and $l$-Lagrangian subbundles to immersed submanifolds as follows.

\begin{definition}
Let $(M,\omega)$ be a multisymplectic manifold of order $k$, and $N \hookrightarrow M$ a submanifold
with canonical embedding $i \colon N \hookrightarrow M$. Let us consider the subbundle
$\Tan i(\Tan N) \subseteq \Tan M$. Then, $N$ is a \textnormal{$l$-isotropic (immersed) submanifold}
(resp., \textnormal{$l$-coisotropic submanifold}, \textnormal{$l$-Lagrangian submanifold}) if
$\Tan i(\Tan N)$ is a $l$-isotropic (resp., $l$-coisotropic, $l$-Lagrangian) subbundle.
\end{definition}

Finally, one has the following characterization of isotropic submanifolds of maximum order.

\begin{lemma}
A submanifold $i \colon N \hookrightarrow M$ is $(k-1)$-isotropic if, and only if, $i^*\omega = 0$.
\end{lemma}

\subsection{First-order jet bundles}
\label{geomset}

In this section we give a short review on jet bundles: definition, some canonical structures and
the concept of ``dual bundle'' (see \cite{book:Saunders89} for details).

\subsubsection*{Definition and local coordinates. Prolongation of sections. Holonomy.}
\label{sec:HOJetBundles}

Let $M$ be an orientable $m$-dimensional smooth manifold with fixed volume form $\eta \in \df^m(M)$,
and let $E \stackrel{\pi}{\longrightarrow} M$ be a bundle with $\dim E = m + n$.
The \textsl{first-order jet bundle} of the projection $\pi$, $J^1\pi$, is the manifold of the
$1$-jets of local sections $\phi \in \Gamma(\pi)$; that is, equivalence classes of local sections
of $\pi$ by the relation of equality on every first-order partial derivative. A point in $J^1\pi$
is denoted by $j^1_x\phi$, where $x \in M$ and $\phi \in \Gamma(\pi)$ is a representative of the
equivalence class. The manifold $J^1\pi$ is endowed with the following natural projections
$$
\begin{array}{rcl}
\pi^1 \colon J^1\pi & \longrightarrow & E \\
j^1_x\phi & \longmapsto & \phi(x)
\end{array}
\quad ; \quad
\begin{array}{rcl}
\bar{\pi}^1 = \pi \circ \pi^1 \colon J^1\pi & \longrightarrow & M \\
j^1_x\phi & \longmapsto & x
\end{array} \, .
$$
The fibers $(\pi^1)^{-1}(u) \subseteq J^1\pi$, with $u \in E$, are denoted $J^1_u\pi$.

Local coordinates in $J^1\pi$ are introduced as follows: let $(x^i)$,
$1 \leqslant i \leqslant m$, be local coordinates in $M$ such that
$\eta = \d^mx = \d x^1 \wedge \ldots \wedge \d x^m$, and $(x^i,u^\alpha)$,
$1 \leqslant \alpha \leqslant n$, local coordinates in $E$ adapted to the bundle structure.
Let $\phi \in \Gamma(\pi)$ be a section with coordinate expression
$\phi(x^i) = (x^i,\phi^\alpha(x^i))$. Then, local coordinates in $J^1\pi$ are
$(x^i,u^\alpha,u_i^\alpha)$, with $1 \leqslant i \leqslant m$ and
$1 \leqslant \alpha \leqslant n$, where
$$
u^\alpha = \phi^\alpha \quad ; \quad u_i^\alpha = \derpar{\phi^\alpha}{x^i} \, .
$$

Using these coordinates, the local expressions of the natural projections are
$$
\pi^1(x^i,u^\alpha,u_i^\alpha) = (x^i,u^\alpha) \quad ;\quad
\bar{\pi}^1(x^i,u^\alpha,u_i^\alpha) = (x^i) \, .
$$

If $\phi \in \Gamma(\pi)$ is a section, we denote the \textsl{prolongation} of $\phi$ to $J^1\pi$
by $j^1\phi \in \Gamma(\bar{\pi}^1)$. In the natural coordinates of $J^1\pi$, if
$\phi(x^i) = (x^i,\phi^\alpha(x^i))$, the prolongation of $\phi$ is given by
$$
j^1\phi(x^i) = \left( x^i,\phi^\alpha,\derpar{\phi^\alpha}{x^i} \right) \, .
$$

\begin{definition}
A section $\psi \in \Gamma(\bar{\pi}^1)$ is \textnormal{holonomic} if $j^{1}(\pi^1 \circ \psi) = \psi$,
that is, if there exists a section $\phi = \pi^1 \circ \psi \in \Gamma(\pi)$ such that $\psi$ is
the prolongation of $\phi$ to $J^1\pi$.
\end{definition}

In natural coordinates, if $\psi \in \Gamma(\bar{\pi}^1)$ is given by
$\psi(x^i) = (x^i,\psi^\alpha,\psi_i^\alpha)$, then the condition for $\psi$ to be holonomic gives
the system of partial differential equations
\begin{equation}\label{eqn:HolonomyConditionSect}
\psi_i^\alpha = \derpar{\psi^\alpha}{x^i} \, ,
\quad 1 \leqslant i \leqslant m \, , \ 1 \leqslant \alpha \leqslant n \, ,
\end{equation}

An alternative characterization of holonomic sections is given in terms of the
\textsl{canonical structure form} of $J^1\pi$.

\begin{definition}
The \textnormal{canonical structure form} of $J^1\pi$ is the $1$-form $\theta$ in $J^1\pi$ with
values in $V(\pi)$ defined by
\begin{equation*}
\theta_{j^1_x\phi}(v) = (d_{\phi(x)}^{\rm v}\phi)(\Tan_{j^1_x\phi}\pi^{1}(v)) \, ,
\end{equation*}
where $v \in \Tan_{j^1_x\phi}J^1\pi$ and $d_{\phi(x)}^{\rm v}\phi$ is the vertical differential of
$\phi$ at $\phi(x) \in E$, and is defined as the map
$d_{\phi(x)}^{\rm v} \phi \colon \Tan_{\phi(x)} E \to \Tan_{\phi(x)}E$ such that
$d_{\phi(x)}^{\rm v}\phi = \Tan_{\phi(x)}\Id - \Tan_{\phi(x)}(\phi \circ \pi)$.
\end{definition}

\begin{proposition}
A section $\psi \in \Gamma(\bar{\pi}^{1})$ is holonomic if, and only if, $\psi^*\theta = 0$.
\end{proposition}

\subsubsection*{Dual bundles}

Let us consider the dual space of $J^1\pi$ as an affine bundle over $E$, which is the set of affine
maps from $J_u^1\pi$ to $(\Lambda^m\Tan^*M)_{\pi(u)}$, with $u \in E$, that is, the set
$$
\bigcup_{u \in E} \textnormal{Aff}(J^1_u\pi,(\Lambda^m(\Tan^*M))_{\pi(u)}) \, .
$$
From \cite{art:Carinena_Crampin_Ibort91} we know that this set is a manifold diffeomorphic to the
smooth vector bundle of $\pi$-semibasic $m$-forms over $E$, $\Lambda^m_2(\Tan^*E)$. This bundle is
called the \textsl{extended dual jet bundle of $\pi$}, and we have the following canonical
projections
\begin{equation*}
\begin{array}{rcl}
\pi_{E} \colon \Lambda^m_2(\Tan^*E) & \longrightarrow & E \\
(u,\omega_u) & \longmapsto & u
\end{array}
\quad ; \quad
\begin{array}{rcl}
\bar{\pi}_{E} \colon \Lambda^m_2(\Tan^*E) & \longrightarrow & M \\
(u,\omega_u) & \longmapsto & \pi(u)
\end{array} \, .
\end{equation*}

Since $\Lambda^m_2(\Tan^*E)$ is a bundle of forms, it is endowed with some canonical forms.
First, the \textsl{Liouville $m$-form}, or \textsl{tautological $m$-form}, is the form
$\Theta \in \df^{m}(\Lambda^m_2(\Tan^*E))$ defined by
$$
\Theta(\omega)(X_1,\ldots,X_m) = \omega(\Tan_\omega\pi_E(X_1),\ldots,\Tan_\omega\pi_E(X_m)) \, ,
$$
where $\omega \in \Lambda_2^m(\Tan^*E)$, and $X_1,\ldots,X_m \in \Tan_{\omega}(\Lambda_2^m(\Tan^*E))$.
As usual, this form satisfies the property $\xi^*\Theta = \xi$ for every $\xi \in \df^{m}(E)$.
From this, the \textsl{Liouville $(m+1)$-form}, or \textsl{canonical multisymplectic $(m+1)$-form},
is the form $\Omega = -\d\Theta \in \df^{m+1}(\Lambda^m_2(\Tan^*E))$.

Local coordinates in $\Lambda^m_2(\Tan^*E)$ are constructed as follows: let $(x^i)$ be local coordinates
in $M$, and $(x^i,u^\alpha)$ coordinates in $E$ adapted to the bundle structure. Then, local coordinates
in $\Lambda^m_2(\Tan^*E)$ are $(x^i,u^\alpha,p,p_\alpha^i)$, where $1 \leqslant i  \leqslant m$,
$1 \leqslant \alpha \leqslant n$. In these coordinates, the canonical projections have the following
local expressions
\begin{equation*}
\pi_E(x^i,u^\alpha,p,p_\alpha^i) = (x^i,u^\alpha) \quad ; \quad
\bar{\pi}_E(x^i,u^\alpha,p,p_\alpha^i) = (x^i) \, .
\end{equation*}
On the other hand, the Liouville $m$ and $(m+1)$-forms have the following local expressions
\begin{equation}\label{eqn:LiouvilleForms}
\Theta = p\d^mx + p^i_\alpha \d u^\alpha \wedge \d^{m-1}x_i \quad ; \quad
\Omega = -\d p \wedge \d^mx - \d p^i_\alpha \wedge \d u^\alpha \wedge \d^{m-1}x_i \, ,
\end{equation}
where $\d^mx = \d x^1 \wedge \ldots \wedge \d x^m$ and $\d^{m-1}x_i = \inn(\partial/\partial x^i)\d^m x$.
It is clear from this coordinate expression that $\Omega$ is a multisymplectic form on
$\Lambda^m_2(\Tan^*E)$.

As $\Lambda^m_2(\Tan^*E)$ is, in fact, a vector bundle over $E$, we can consider its quotient by
any vector subbundle. The \textsl{reduced dual jet bundle of $\pi$}, denoted $J^1\pi^*$,
is the quotient of the extended dual jet bundle, $\Lambda^{m}_{2}(\Tan^*E)$, by constant affine
transformations along the fibers of $\pi^1$, and is diffeomorphic to the quotient bundle
$\Lambda^{m}_{2}(\Tan^*E)/\Lambda^{m}_{1}(\Tan^*E)$. The natural quotient map is
$\mu \colon \Lambda^{m}_{2}(\Tan^*E) \to J^{1}\pi^*$.

It can be proved that $J^{1}\pi^*$ may be endowed with the structure of a smooth manifold and,
moreover, $\mu \colon \Lambda^{m}_{2}(\Tan^*E) \to J^{1}\pi^*$ is a smooth vector bundle of rank $1$.
In addition, we have the canonical projections $\pi_E^r \colon J^{1}\pi^* \to E$ and
$\bar{\pi}_E^r \colon J^{1}\pi^* \to M$.

Finally, adapted coordinates $(x^i,u^\alpha)$ in $E$ induce coordinates $(x^i,u^\alpha,p_\alpha^i)$
in $J^{1}\pi^*$ such that the coordinate expression of the natural quotient map is
\begin{equation*}
\mu(x^i,u^\alpha,p,p_\alpha^{i}) = (x^i,u^\alpha,p_\alpha^{i}) \, ,
\end{equation*}
where $(x^i,u^\alpha,p,p_\alpha^{i})$ are the induced coordinates in $\Lambda^{m}_2(\Tan^*E)$.
In these coordinates, the natural projections are given by
\begin{equation*}
\pi_E^r(x^i,u^\alpha,p_\alpha^{i}) = (x^i,u^\alpha) \quad ; \quad
\bar{\pi}_E^r(x^i,u^\alpha,p_\alpha^{i}) = (x^i) \, .
\end{equation*}

\subsection{Multivector fields}
\label{sec:MultivectorFields}

In this section we give a short review on multivector fields and their relation with integrable
distributions (see \cite{art:Echeverria_Munoz_Roman98} for details).

\subsubsection*{Locally decomposable multivector fields. Integrability conditions}

A \textsl{multivector field of degree $k$}, or \textsl{$k$-multivector field}, on a $m$-dimensional
smooth manifold $M$ is a section of the bundle $\Lambda^{k}(\Tan M) \to M$, that is, a skew-symmetric
contravariant tensor. The set of all multivector fields of degree $k$ in $M$ is denoted $\vf^{k}(M)$.

In general, given a $k$-multivector field $\X \in \vf^k(M)$, for every $p \in M$ there exists an
open neighborhood $U_p \subseteq M$ and $X_1,\ldots,X_r \in \vf(U_p)$ such that
$$
\X = \sum_{1 \leqslant i_1 < \ldots < i_k \leqslant r} f^{i_1\dots i_k} X_{i_1} \wedge \ldots \wedge X_{i_k} \, ,
$$
with $f^{i_1\dots i_k} \in \Cinfty(U_p)$ and $k \leqslant r \leqslant \dim M$. If for every $p$ we
have $r = k$, that is, there exists an open neighborhood $U_p \subseteq M$ and
$X_1,\ldots,X_k \in \vf(U_p)$ such that $\X = X_1 \wedge \ldots \wedge X_k$ on $U_p$, then we say
that the multivector field $\X$ is \textsl{locally decomposable}.

Let $\D$ be a $k$-dimensional distribution in $M$, that is, a $k$-dimensional subbundle of $\Tan M$.
It is clear that sections of $\Lambda^k\D \to M$ are $k$-multivector fields in $M$, and that the
existence of a non-vanishing global section of $\Lambda^k\D \to M$ is equivalent to the orientability
of the distribution $\D$. Then, we say that a non-vanishing multivector field $\X \in \vf^{k}(M)$ and a
$k$-dimensional distribution $\D \subset \Tan M$ are \textsl{locally associated} if there exists
a connected open set $U \subseteq M$ such that $\restric{\X}{U}$ is a section of $\restric{\Lambda^{k}\D}{U}$.

As a consequence of this we can introduce an equivalence relation on the set of non-vanishing
$k$-multivector fields in $M$ as follows: two $k$-multivector fields $\X,\X^\prime \in \vf^{k}(M)$
are related if, and only if, they are both locally associated, on the same connected open set
$U \subseteq M$, with the same distribution $\D$. In addition, in this case there exists a
non-vanishing function $f \in \Cinfty(U)$ such that $\X^\prime = f\X$ on $U$. The equivalence
classes of this quotient set will be denoted by $\{ \X \}_U$. Then, one can prove that there is
a bijective correspondence between the set of $k$-dimensional orientable distributions
$\D \subseteq \Tan M$ and set of equivalence classes $\{ \X \}_M$ of non-vanishing, locally
decomposable $k$-multivector fields in $M$.

If $\X \in \vf^{k}(M)$ is a non-vanishing, locally decomposable $k$-multivector field and
$U \subseteq M$ is a connected open set, then the distribution associated to the equivalence
class $\{\X\}_U$ will be denoted by $\D_U(\X)$. If $U = M$, then we write simply $\D(\X)$.

A non-vanishing, locally decomposable multivector field $\X \in \vf^k(M)$ is said to be
\textsl{integrable} (resp. \textsl{involutive}) if  its associated distribution $\D_U(\X)$ is
integrable (resp. involutive). It is clear then that if $\X \in \vf^k(M)$ is integrable (resp.
involutive), then so is every other in its equivalence class $\{ \X \}$, and all of them have the
same integral manifolds. Moreover, \textsl{Frobenius theorem} allows us to state that a non-vanishing
and locally decomposable multivector field is integrable if, and only if, it is involutive.
Nevertheless, in many applications we have locally decomposable multivector fields $\X \in \vf^k(M)$
which are not integrable in $M$, but integrable in a submanifold of $M$. A (local) algorithm for
finding this submanifold has been developed \cite{art:Echeverria_Munoz_Roman98}.

\subsubsection*{Multivector fields in fiber bundles and jet bundles. Holonomy condition}

We are interested in the particular situation of a fiber bundle and, more precisely, of jet bundles. 

Let $\pi \colon E \to M$ be a fiber bundle, with $\dim M = m$ and $\dim E = m + n$. A multivector
field $\X \in \vf^{m}(M)$ is said to be \textsl{$\pi$-transverse} if at every point $u \in E$ we have
$(\inn(\X)(\pi^*\omega))_y \neq 0$ for every $\omega \in \df^{m}(M)$ satisfying $\omega(\pi(y)) \neq 0$.
It can be proved that if $\X$ is integrable, then the $\pi$-transverse condition is equivalent to
requiring the integral manifolds of $\X$ to be local sections of $\pi$.
In this case, if $\phi \colon U \subseteq M \to E$ is a local section with $\phi(x) = u$
and $\phi(U)$ is the integral manifold of $\X$, then $\Tan_u(\Im\phi) = \D_u(\X)$.

Now, let us consider the first-order jet bundle of $\pi$, $J^1\pi$. A multivector field
$\X \in \vf^m(J^{1}\pi)$ is \textsl{holonomic} if $\X$ is integrable, $\bar{\pi}^1$-transverse,
and the integral sections of $\X$ are holonomic.

As in mechanics, the holonomy of a multivector field may be characterized using the geometry of
$J^1\pi$. First, a $\pi^1$-transverse and locally decomposable multivector field
$\X \in \vf^{m}(J^1\pi)$ is said to be \textsl{semi-holonomic}, or a \textsl{SOPDE multivector field}
if, and only if, $\inn(\theta)\X = 0$, where $\theta$ is the canonical structure form in $J^1\pi$.
Then, the relation between integrable, holonomic and semi-holonomic multivector fields in $J^1\pi$
is given by the following result from \cite{art:Echeverria_Munoz_Roman98}.

\begin{theorem}\label{thm:Holonomic_VS_Integrability&Semiholonomic}
A multivector field $\X \in \vf^{m}(J^1\pi)$ is holonomic if, and only if, it is integrable and
semi-holonomic.
\end{theorem}

In natural coordinates, let $\X \in \vf^{m}(J^1\pi)$ be a locally decomposable and
$\bar{\pi}^1$-transverse multivector field. From the results in \cite{art:Echeverria_Munoz_Roman98},
this multivector field $\X$ may be chosen to have the following coordinate expression
\begin{equation*}
\X = f \bigwedge_{j=1}^{m} X_j = f \bigwedge_{j=1}^{m}
\left( \derpar{}{x^j} + f_j^\alpha\derpar{}{u^\alpha} + F_{j,i}^\alpha\derpar{}{u_i^\alpha} \right) \, ,
\end{equation*}
with $f$ being a non-vanishing local function. Then, the condition for $\X$ to be semi-holonomic
gives the $mn$ equations $f_j^\alpha = u_j^\alpha$. In addition, from the results in
\cite{art:Echeverria_Munoz_Roman98}, we know that the necessary and sufficient condition for a
locally decomposable multivector field to be integrable is that its associated distribution is
involutive, which is equivalent to requiring the $m(m-1)/2$ conditions $[ X_j,X_k ] = 0$, with
$1 \leqslant j < k \leqslant m$. In coordinates, these gives the following system of $nm(m^2-1)/2$
partial differential equations for the component functions $F_{j,i}^\alpha$
\begin{equation}\label{eqn:Integrability&SemiHolonomicity}
F_{j,k}^\alpha - F_{k,j}^\alpha = 0 \quad ; \quad
\derpar{F_{k,i}^\alpha}{x^j} + u_j^\beta \derpar{F_{k,i}^\alpha}{u^\beta} + F_{j,l}^\beta \derpar{F_{k,i}^\alpha}{u_l^\beta}
- \derpar{F_{j,i}^\alpha}{x^k} - u_k^\beta \derpar{F_{j,i}^\alpha}{u^\beta} - F_{k,l}^\beta \derpar{F_{j,i}^\alpha}{u_l^\beta} = 0 \, .
\end{equation}

\begin{remark}
It is important to point out that a locally decomposable, $\bar{\pi}^1$-transverse and semi-holonomic
multivector field $\X$ may not be holonomic, since the SOPDE condition is not a sufficient nor
necessary condition for the multivector field to be integrable. On the other hand, the integrability
of a multivector field does not imply that the integral sections are holonomic: as in mechanics, a
multivector field may admit integral sections through every point in $J^1\pi$, but these integral
sections may not be projectable to the base manifold.
\end{remark}

\subsubsection*{Relation with jet fields}

Let $\pi \colon E \to M$ be a fiber bundle, with $\dim M = m$ and $\dim E = m + n$, and $J^1\pi$
the first-order jet bundle of $\pi$.

\begin{definition}
A \textnormal{jet field} in $E$ is a global section of the projection $\pi^1 \colon J^1\pi \to E$.
\end{definition}

It is proved in \cite{book:Saunders89} that there is a bijective correspondence between jet fields in
$E$ and connections $\nabla \in \Lambda^1_1(T^*E) \otimes \vf(E)$. Therefore, there is a bijective
correspondence between jet fields in $E$ and distributions in $E$. We denote $\D(\Psi)$ the unique
distribution in $E$ associated to the jet field $\Psi \in \Gamma(\pi^1)$. This enables us to give
the following definitions.

\begin{definition}
Let $\Psi \in \Gamma(\pi^1)$ be a jet field and $\D(\Psi)$ its associated distribution in $E$.
\begin{enumerate}
\item $\Psi$ is said to be \textnormal{orientable} if, and only if, $\D(\Psi)$ is an orientable
distribution in $E$. In particular, if $M$ is orientable, then every jet field is also orientable.
\item $\Psi$ is said to be \textnormal{integrable} if, and only if, $\D(\Psi)$ is an integrable
distribution.
\item A section $\phi \in \Gamma(\pi)$ is an \textnormal{integral section} of $\Psi$ if, and only
if, $\Psi \circ \phi = j^1\phi$. In particular, $\Psi$ is integrable if, and only if, it admits
integral sections through every point of $E$.
\end{enumerate}
\end{definition}

With these notations, the relation between multivector fields and jet fields is given by the
following result, stated in \cite{art:Echeverria_Munoz_Roman98}.

\begin{theorem}
There is a bijective correspondence between the set of orientable jet fields $\Psi \in \Gamma(\pi^1)$
and the set of equivalence classes of locally decomposable and $\pi$-transverse multivector fields
$\{ \X \} \subseteq \vf^{m}(E)$. They are characterized by the fact that $\D(\Psi) = \D(\X)$. In
addition, the orientable jet field $\Psi$ is integrable if, and only if, so is every $\X$  in the
equivalence class.
\end{theorem}

\section{The Hamilton-Jacobi problem in the Lagrangian formalism}
\label{sec:LagrangianFormalism}

The geometrical setting for the Lagrangian formalism for multisymplectic field theories is the
following (see, for instance, \cite{AA-80,EMR-96,art:Echeverria_Munoz_Roman98,Gc-74,GS-73,GIMMSY-mm,
LMM-96,art:Roman09} for more details). Let $\pi \colon E \to M$ be a fiber bundle modeling the
configuration space, where $M$ is a $m$-dimensional orientable smooth manifold with fixed volume form
$\eta \in \df^{m}(M)$, and $\dim E = m + n$. Let $\Lag \in \df^{m}(J^1\pi)$ be a Lagrangian density
containing the physical information of the theory, that is, a $\bar{\pi}^1$-semibasic $m$-form. We
denote $L \in \Cinfty(J^1\pi)$ the function satisfying $\Lag = L(\bar{\pi}^1)^*\eta$, which we call
the Lagrangian function associated to $\Lag$ and $\eta$. Using the canonical vertical endomorphism
$\nu \in \Gamma(\Tan^*J^1\pi \otimes_{J^1\pi} \Tan M \otimes_{J^1\pi} V(\pi^1))$, the Cartan forms
$\Theta_\Lag = \inn(\nu)\d\Lag + \Lag \in \df^m(J^1\pi)$ and
$\Omega_\Lag = -\d\Theta_\Lag \in \df^{m+1}(J^1\pi)$ are constructed, with coordinate expressions
\begin{align}
&\Theta_\Lag = \derpar{L}{u_i^\alpha}\d u^\alpha \wedge \d^{m-1}x_i -
\left( \derpar{L}{u_i^\alpha}u_i^\alpha - L \right)\d^{m}x \, , \label{eqn:CartanMFormLocal} \\
&\begin{array}{l}
\displaystyle \Omega_{\Lag} =
\derpars{L}{u_i^\alpha}{u^\beta} \, \d u^\alpha \wedge \d u^\beta \wedge \d^{m-1}x_i
+ \derpars{L}{u_i^\alpha}{u_j^\beta} \, \d u^\alpha \wedge \d u_j^\beta \wedge \d^{m-1}x_i \\[12pt]
\displaystyle \qquad\qquad + \left( u_i^\alpha \derpars{L}{u_i^\alpha}{u^\beta} - \derpar{L}{u^\beta} \right) \d u^\beta \wedge \d^mx
+ u_i^\alpha\derpars{L}{u_i^\alpha}{u_j^\beta} \, \d u_j^\beta \wedge \d^mx \, .
\end{array}\label{eqn:MultisymplecticCartanFormLocal}
\end{align}

Then, the Lagrangian problem for first-order classical field theories is the following: \textit{to find a
$m$-dimensional, $\bar{\pi}^{1}$-transverse and integrable distribution $\D_\Lag$ in $J^1\pi$ such that
the integral sections of $\D_\Lag$ are prolongations of sections $\phi \in \Gamma(\pi)$ satisfying}
\begin{equation}\label{eqn:LagFieldEqSect}
(j^1\phi)^*\inn(X)\Omega_\Lag = 0 \, , \ \mbox{\textit{for every} } X \in \vf(J^1\pi) \, .
\end{equation}
If the Lagrangian density is regular, then the Cartan $(m+1)$-form $\Omega_\Lag$ is multisymplectic,
and then there exists such a distribution, although it is not necessarily integrable. In the following
we assume that the Lagrangian density $\Lag$ is regular, and that the distribution $\D_\Lag$ is,
in addition, integrable.

From the results in Section \ref{sec:MultivectorFields}, this distribution $\D_\Lag$ is associated
with a class of holonomic multivector fields $\{\X_\Lag\} \subseteq \vf^{m}(J^1\pi)$ satisfying
the equation
\begin{equation}\label{eqn:LagFieldEqMVF}
\inn(\X_\Lag)\Omega_\Lag = 0 \, , \ \mbox{for every } \X_\Lag \in \left\{ \X_\Lag \right\} \, .
\end{equation}
The same comments apply in the regular case: if the Lagrangian density is regular, then there exists a
class of multivector fields $\{\X_\Lag\} \subset \vf^{m}(J^1\pi)$ solution to equation
\eqref{eqn:LagFieldEqMVF} which is $\bar{\pi}^1$-transverse and SOPDE, but not necessarily integrable.
In the following we assume that the Lagrangian density is regular and that every multivector field
in the class is, in addition, integrable. This class is denoted by $\{ \X_\Lag \}$ along this work.

\subsection{The generalized Lagrangian Hamilton-Jacobi problem}
\label{sec:LagGenHJProblem}

Following the patterns in \cite{art:Carinena_Gracia_Marmo_Martinez_Munoz_Roman06}, we first
state a generalized version of the Hamilton-Jacobi problem in the Lagrangian formalism.

\begin{definition}\label{def:LagGenHJDef}
The \textnormal{generalized Lagrangian Hamilton-Jacobi problem} consists in finding a jet field
$\Psi \in \Gamma(\pi^1)$ and a $m$-dimensional and integrable distribution $\D$ in $E$ such that if
$\gamma \in \Gamma(\pi)$ is an integral section of $\D$, then
$\Psi \circ \gamma \in \Gamma(\bar{\pi}^1)$ is an integral section of $\D_\Lag$, that is,
\begin{equation}\label{eqn:LagGenHJDefDist}
\Tan_{u}\Im(\gamma) = \D_{u} \ \ \forall u \in \Im(\gamma) \Longrightarrow
\Tan_{\bar{u}}\Im(\Psi \circ \gamma) = (\D_\Lag)_{\bar{u}}
\ \ \forall \ \bar{u} \in \Im(\Psi \circ \gamma) \, .
\end{equation}
\end{definition}

From the results in Section \ref{sec:MultivectorFields}, since both $\D$ and $\D_\Lag$ are
associated with their corresponding classes of multivector fields, the problem can be stated
equivalently in terms of multivector fields as the search of a jet field
$\Psi \in \Gamma(\pi^1)$ and a class of locally decomposable and integrable multivector fields
$\{ \X \} \subseteq \vf^{m}(E)$ such that if $\gamma \in \Gamma(\pi)$ is an integral section of
every multivector field $\X \in \{ \X \}$ then $\Psi \circ \gamma \in \Gamma(\bar{\pi}^1)$ is an
integral section of every multivector field $\X_\Lag \in \{\X_\Lag\}$ solution to equation
\eqref{eqn:LagFieldEqMVF}, that is,
\begin{equation}\label{eqn:LagGenHJDefMVF}
\X \circ \gamma = \Lambda^m\dot{\gamma} \ \ \forall \X \in \{\X\} \Longrightarrow
\X_\Lag \circ (\Psi \circ \gamma) = \Lambda^m(\dot{\overline{\Psi \circ \gamma}})
\ \ \forall \X_\Lag \in \{\X_\Lag\} \, ,
\end{equation}
where $\Lambda^m\dot{\gamma} \colon M \to \Lambda^m(\Tan E)$ denotes the canonical lift of $\gamma$ to
$\Lambda^m(\Tan E)$. In the following we denote by $\{ \X \} \subseteq \vf^{m}(E)$ the class of locally
decomposable and integrable multivector fields associated with the integable distribution $\D$ in $E$.
The diagram illustrating this equivalent formulation of the generalized Lagrangian Hamilton-Jacobi
problem is the following
$$
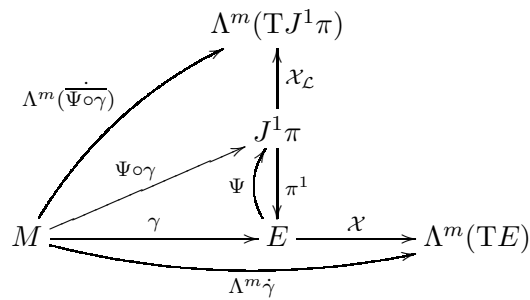
\xymatrix{
\ & \ & \Lambda^m(\Tan J^1\pi) & \ \\
\ & \ & J^1\pi \ar[u]_-{\X_\Lag} \ar[d]^{\pi^1} & \ \\
M \ar[rr]^-{\gamma} \ar[urr]^-{\Psi \circ \gamma} \ar@/^1pc/[uurr]^-{\Lambda^m(\dot{\overline{\Psi \circ \gamma}})} \ar@/_1pc/[rrr]_{\Lambda^m\dot{\gamma}} & \ & E \ar@/^.75pc/[u]^{\Psi} \ar[r]^-{\X} & \Lambda^m(\Tan E)
}
$$
where the interpretation is: \textit{if the lower diagram formed by $\gamma$, $\X$ and
$\Lambda^m\dot{\gamma}$ is commutative for every $\X \in \{ \X \}$, then the upper diagram formed
by $\Psi \circ \gamma$, $\X_\Lag$ and $\Lambda^m(\dot{\overline{\Psi \circ \gamma}})$ is also
commutative for every $\X_\Lag \in \{ \X_\Lag \}$.}

\begin{remark}
Since the section $\Psi \circ \gamma \in \Gamma(\bar{\pi}^1)$ is an integral section of $\D_\Lag$
(or, equivalently, an integral section of the associated class of holonomic multivector fields),
in particular it must satisfy equation \eqref{eqn:LagFieldEqSect}, that is,
$$
(\Psi \circ \gamma)^*\inn(X)\Omega_\Lag = 0 \, , \ \mbox{for every } X \in \vf(J^1\pi) \, .
$$
Nevertheless, observe that the action of the $m$-form
$(\Psi \circ \gamma)^*\inn(X)\Omega_\Lag \in \df^{m}(M)$ on $m$ tangent vectors $v_i \in \Tan_xM$,
with $x \in M$, is defined as
$$
\left((\Psi \circ \gamma)^*\inn(X)\Omega_\Lag\right)_{x}(v_1,\ldots,v_m) =
(\Omega_\Lag)_{\Psi(\gamma(x))}(X(\Psi(\gamma(x))),\Tan_x(\Psi \circ \gamma)(v_1),\ldots,\Tan_x(\Psi \circ \gamma)(v_m)) \, ,
$$
from where we observe that
$X(\Psi(\gamma(x))) \in \Tan_{\Psi(\gamma(x))} \Im(\Psi \circ \gamma) \subset \Tan_{\Psi(\gamma(x))}\Im(\Psi)$,
that is, the vector field $X$ is tangent to the submanifold $\Im(\Psi) \hookrightarrow J^1\pi$.
Therefore, in this particular situation, equation \eqref{eqn:LagFieldEqSect} is equivalent to
\begin{equation}\label{eqn:LagFieldEqSectWithHJSolution}
(\Psi \circ \gamma)^*\inn(X)\Omega_\Lag = 0 \, ,
\ \mbox{for every } X \in \vf(J^1\pi) \mbox{ tangent to } \Im(\Psi) \, .
\end{equation}
\end{remark}

\begin{remark}\label{rem:LagGenHJDefJetFields}
Since every integral section of the distribution $\D_\Lag$ is the prolongation of a section of $\pi$,
this holds, in particular, for the section $\Psi \circ \gamma$, and we have $\Psi \circ \gamma = j^1\phi$
for some $\phi \in \Gamma(\pi)$. Now, composing this last equality with the natural projection $\pi^1$,
we obtain $\gamma = \phi$. Then, replacing $\phi$ by $\gamma$ in the previous expression, we have
$\Psi \circ \gamma = j^1\gamma$, from where we deduce that if $\D$ is an integrable distribution, then
the jet field $\Psi \in \Gamma(\pi^1)$ is integrable, and every integral section of $\D$ is an integral
section of $\Psi$. Moreover, this enables us to reformulate the generalized Lagrangian Hamilton-Jacobi
problem as follows:

\noindent\textit{The \textnormal{generalized Lagrangian Hamilton-Jacobi problem} consists in finding an
integrable jet field $\Psi \in \Gamma(\pi^1)$ such that if $\gamma \in \Gamma(\pi)$ is an
integral section of $\Psi$, then $j^1\gamma \in \Gamma(\bar{\pi}^1)$ is an integral section of $\D_\Lag$.}

\noindent Nevertheless, we stick to the statement in Definition \ref{def:LagGenHJDef}, or the equivalent
formulation given in terms of multivector fields, in order to give several equivalent conditions to
being a solution to the generalized Lagrangian Hamilton-Jacobi problem.
\end{remark}

It is clear from this last remark that the distribution $\D$ in $E$, the jet field
$\Psi \in \Gamma(\pi^1)$ and the distribution $\D_\Lag$ in $J^1\pi$ are closely related.
In fact, we have the following result.

\begin{proposition}\label{prop:LagGenHJRelatedDist}
The jet field $\Psi \in \Gamma(\pi^1)$ and the distribution $\D$ in $E$ satisfy condition
\eqref{eqn:LagGenHJDefDist} if, and only if, $\D$ and $\D_\Lag$ are $\Psi$-related, that is, for
every $\X_\Lag \in \{ \X_\Lag \}$ (resp., for every $\X \in \{ \X \}$) there exists $\X \in \{ \X \}$
(resp., $\X_\Lag \in \{ \X_\Lag \}$) such that $\X_\Lag \circ \Psi = \Lambda^m\Tan\Psi \circ \X$.
\end{proposition}
\begin{proof}
We prove this result in terms of the associated classes of multivector fields. Let
$\gamma \in \Gamma(\pi)$ be an integral section of $\D$, which is equivalent to $\gamma$ being an
integral section of every $\X \in \{ \X \}$, and let $\X_\Lag \in \{ \X_\Lag \}$ be a representative
of the equivalence class. Then we have
$$
\X_\Lag \circ \Psi \circ \gamma = \Lambda^m(\dot{\overline{\Psi \circ \gamma}})
= \Lambda^m\Tan\Psi \circ \Lambda^m\dot{\gamma} = \Lambda^m\Tan\Psi \circ \X \circ \gamma \, ,
$$
where the multivector field $\X \in \{ \X \}$ in the last equality exists since $\gamma$ is an
integral section of $\D$. Then, since $\X$ is integrable, it admits integral sections through every
point in $E$, and therefore we have proved that for every $\X_\Lag \in \{ \X_\Lag \}$ there exists
$\X \in \{ \X \}$ such that $\X_\Lag \circ \Psi = \Lambda^m\Tan\Psi \circ \X$. Reversing this
reasoning we prove that for every $\X \in \{ \X \}$ there exists $\X_\Lag \in \{ \X_\Lag \}$ such
that $\X_\Lag \circ \Psi = \Lambda^m\Tan\Psi \circ \X$. Therefore, the distributions $\D$ and
$\D_\Lag$ are $\Psi$-related.

Conversely, let us suppose that $\D$ and $\D_\Lag$ are $\Psi$-related. Then, if
$\gamma \in \Gamma(\pi)$ is an integral section of $\D$ and $\X_\Lag \in \{ \X_\Lag \}$
is a representative of the equivalence class, we have
$$
\X_\Lag \circ \Psi \circ \gamma = \Lambda^m\Tan\Psi \circ \X \circ \gamma
= \Lambda^m\Tan\Psi \circ \Lambda^m\dot{\gamma} = \Lambda^m(\dot{\overline{\Psi \circ \gamma}}) \, ,
$$
where $\X \in \{ \X \}$ in the first equality is any multivector field in $\{ \X \}$ which is
$\Psi$-related to the given $\X_\Lag$. That is, the jet field $\Psi$ and the distribution $\D$
satisfy condition \eqref{eqn:LagGenHJDefDist}.
\end{proof}

A straightforward consequence of this last result is the following.

\begin{corollary}\label{corol:LagGenHJAssociatedMVF}
If the jet field $\Psi \in \Gamma(\pi^1)$ and the distribution $\D$ in $E$ satisfy condition
\eqref{eqn:LagGenHJDefDist}, then for every $\X \in \{ \X \}$ there exists a multivector field
$\X_\Lag \in \{ \X_\Lag \}$ such that $\X$ is given by
$$
\X = \Lambda^m\Tan\pi^1 \circ \X_\Lag \circ \Psi \, .
$$
\end{corollary}
\begin{proof}
Let $\X \in \{ \X \}$ be an arbitrary representative of the equivalence class. From Proposition
\ref{prop:LagGenHJRelatedDist} we know that if $\Psi$ and $\D$ satisfy condition
\eqref{eqn:LagGenHJDefDist}, then there exists a multivector field $\X_\Lag \in \{ \X_\Lag \}$ which
is $\Psi$-related to the given $\X$, that is, $\X_\Lag \circ \Psi = \Lambda^m\Tan\Psi \circ \X$.
Then, composing both sides of this last equality with the map
$\Lambda^m\Tan\pi^1 \colon \Lambda^m(\Tan J^1\pi) \to \Lambda^m(\Tan E)$, and bearing in mind that
$\Psi \in \Gamma(\pi)$, we obtain
$$
\Lambda^m\Tan\pi^1 \circ \X_\Lag \circ \Psi = \Lambda^m\Tan\pi^1 \circ \Lambda^m\Tan\Psi \circ \X
= \Lambda^m\Tan (\pi^1 \circ \Psi) \circ \X = \X \, .
$$
\end{proof}

That is, every multivector field  $\X \in \{ \X \}$ is completely determinated by the jet field
$\Psi \in \Gamma(\pi^1)$ and some multivector field $\X_\Lag \in \{ \X_\Lag \}$, and it is called the
\textsl{multivector field associated to $\Psi$ and $\X_\Lag$}. The diagram illustrating this situation
is the following:
$$
\xymatrix{
J^1\pi \ar[rr]^-{\X_\Lag} \ar[dd]^{\pi^1} & \ & \Lambda^m(\Tan J^1\pi) \ar[dd]_{\Lambda^m\Tan\pi^1} \\
\ & \ & \ \\
E \ar@/^1pc/[uu]^{\Psi} \ar[rr]_-{\X} & \ & \Lambda^m(\Tan E) \ar@/_1pc/[uu]_{\Lambda^m\Tan\Psi}
}
$$
In particular, the distribution $\D$ is completely determinated by the jet field
$\Psi \in \Gamma(\pi^1)$ and the distribution $\D_\Lag$ as
$\D_{\pi^1(\bar{u})} = \Tan_{\bar{u}}\pi^1((\D_\Lag)_{\bar{u}})$ for every
$\bar{u} \in \Im(\Psi) \hookrightarrow J^1\pi$, or
$\D = \Tan\pi^1(\restric{\D_\Lag}{\Im(\Psi)})$, and it is called the
\textsl{distribution associated to $\Psi$}.

\begin{remark}
From Corollary \ref{corol:LagGenHJAssociatedMVF} we deduce that if the jet field $\Psi$ is integrable,
then so is every multivector field $\X \in \{ \X \}$, and hence the distribution $\D$. Moreover,
taking into account the remark in page \pageref{rem:LagGenHJDefJetFields}, we deduce that the jet
field $\Psi$ is integrable if, and only if, the distribution $\D$ is integrable, and they have the
same integral sections. Hence, $\Psi$ and $\D$ are associated in the sense of
\cite{art:Echeverria_Munoz_Roman98} (Section \ref{sec:MultivectorFields}), that is, they define the
same horizontal subbundle of $\Tan E$.
\end{remark}

Taking into account Corollary \ref{corol:LagGenHJAssociatedMVF} it is clear that the search for a
jet field $\Gamma(\pi^1)$ and a distribution $\D$ in $E$ satisfying condition
\eqref{eqn:LagGenHJDefDist} is equivalent to the search of a jet field $\Psi$ satisfying the same
condition with the associated distribution $\Tan\pi^1(\restric{\D_\Lag}{\Im(\Psi)}) \subseteq \Tan E$.
Therefore, we can give the following definition.

\begin{definition}
A \textnormal{solution to the generalized Lagrangian Hamilton-Jacobi problem} is an integrable
jet field $\Psi \in \Gamma(\pi^1)$ such that if $\gamma \in \Gamma(\pi)$ is an integral section
of the $m$-dimensional and integrable distribution $\Tan\pi^1(\restric{\D_\Lag}{\Im(\Psi)})$ in $E$,
then $\Psi \circ \gamma \in \Gamma(\bar{\pi}^1)$ is an integral section of $\D_\Lag$.
\end{definition}

Now we state the following characterizations for a jet field to be a solution to the generalized
Lagrangian Hamilton-Jacobi problem.

\begin{proposition}\label{prop:LagGenHJEquivalences}
Let $\Psi \in \Gamma(\pi^1)$ be an integrable jet field. Then, the following statements are
equivalent.
\begin{enumerate}
\item $\Psi$ is a solution to the generalized Lagrangian Hamilton-Jacobi problem.
\item The distribution $\D_\Lag$ in $J^1\pi$ is tangent to the submanifold
$\Im(\Psi) \hookrightarrow J^1\pi$, that is, $(\D_\Lag)_{\bar{u}} \subseteq \Tan_{\bar{u}}\Im(\Psi)$
for every $\bar{u} \in \Im(\Psi)$.
\item $\Psi$ satisfies the equation
$$
\gamma^*\inn(Y)(\Psi^*\Omega_\Lag) = 0 \, , \ \mbox{for every } Y \in \vf(E) \, ,
$$
where $\gamma \in \Gamma(\pi)$ is an integral section of the associated distribution
$\Tan\pi^1(\restric{\D_\Lag}{\Im(\Psi)})$.
\end{enumerate}
\end{proposition}
\begin{proof} \

\noindent ($1 \Longleftrightarrow 2$)
Assume that $\Psi \in \Gamma(\pi^1)$ is a solution to the generalized Lagrangian Hamilton-Jacobi
problem, and let $\gamma \in \Gamma(\pi)$ be an integral section of the integrable distribution
$\D = \Tan\pi^1(\restric{\D_\Lag}{\Im(\Psi)})$. Then, since $\Psi$ and $\D$ satisfy condition
\eqref{eqn:LagGenHJDefDist} and it is clear that $\Tan_{\bar{u}} \Im(\Psi \circ \gamma) \subseteq
\Tan_{\bar{u}} \Im(\Psi)$ for every $\bar{u} \in \Im(\Psi \circ \gamma)$, we have
$$
(\D_\Lag)_{\bar{u}} = \Tan_{\bar{u}}\Im(\Psi \circ \gamma) \subseteq \Tan_{\bar{u}}\Im(\Psi)
\ \ \forall \ \bar{u} \in \Im(\Psi \circ \gamma) \, .
$$
Finally, since $\Psi$ is integrable, for every $\bar{u} \in \Im(\Psi)$ there exists an integral
section $\gamma \in \Gamma(\pi)$ such that $\bar{u} \in \Im(\Psi \circ \gamma)$, and therefore we
have proved $(\D_\Lag)_{\bar{u}} \subseteq \Tan_{\bar{u}}\Im(\Psi)$ for every $\bar{u} \in \Im(\Psi)$.

For the converse, assume that the distribution $\D_\Lag$ in $J^1\pi$ is tangent to the submanifold
$\Im(\Psi) \hookrightarrow J^1\pi$, that is, we have
$(\D_\Lag)_{\bar{u}} \subseteq \Tan_{\bar{u}}\Im(\Psi)$ for every $\bar{u} \in \Im(\Psi)$,
which is equivalent to $(\D_\Lag)_{\Psi(u)} \subseteq \Tan_{\Psi(u)}\Im(\Psi)$ for every $u \in E$.
We deduce from this that for every $w \in (\D_\Lag)_{\Psi(u)}$ there exists a $v_u \in \Tan_uE$ such
that $w = \Tan_u\Psi(v_u)$. Hence, for every $u \in E$ we define a $m$-dimensional subspace
$\D_u \subseteq \Tan_uE$ as follows
$$
\D_u = \left\{ v \in \Tan_u E \mid w = \Tan_u\Psi(v) \, , \ w \in (\D_\Lag)_{\Psi(u)} \right\} \, .
$$
Then, we define the $m$-dimensional distribution $\D$ in $E$ as $\D = \bigcup_{u \in E}\D_u$. It is clear
that $\D$ is a smooth and integrable distribution, since it satisfies
$\D = \Tan\pi^1(\restric{\D_\Lag}{\Im(\Psi)})$, and both $\D_\Lag$ and $\Psi$ are smooth and integrable.
Now, let $\gamma \in \Gamma(\pi)$ be an integral section of $\D$. Then, by definition of $\D$, the
condition for $\gamma$ to be an integral section of $\D$ gives
$$
\Tan_u\Im(\gamma) = \Tan_{\Psi(u)}\pi^1((\D_\Lag)_{\Psi(u)}) \ \ \forall u \in \Im(\gamma) \, .
$$
Composing this equality with the map $\Tan_{u} \Psi \colon \Tan_u E \to \Tan_{\Psi(u)} J^1\pi$,
and bearing in mind that $\Psi \circ \pi^1 = \Id_{\Im(\Psi)}$, we have
$$
(\D_\Lag)_{\Psi(u)} = \Tan_u\Psi(\Tan_u\Im(\gamma)) = \Tan_{\Psi(u)}\Im(\Psi \circ \gamma)
\ \ \forall u \in \Im(\gamma) \, ,
$$
which is clearly equivalent to
$$
(\D_\Lag)_{\bar{u}} = \Tan_{\bar{u}}\Im(\Psi \circ \gamma) \ \ \forall \bar{u} \in \Im(\Psi \circ \gamma) \, ,
$$
that is, to condition \eqref{eqn:LagGenHJDefDist}. Thus, $\Psi$ is a solution to the generalized
Lagrangian Hamilton-Jacobi problem.

\noindent ($1 \Longleftrightarrow 3$)
Let $\Psi \in \Gamma(\pi^1)$ be a jet field solution to the generalized Lagrangian Hamilton-Jacobi
problem, and $\gamma \in \Gamma(\pi)$ an integral section of $\Tan\pi^1(\restric{\D_\Lag}{\Im(\Psi)})$.
Then the section $\Psi \circ \gamma \in \Gamma(\bar{\pi}^1)$ is an integral section of $\D_\Lag$.
In particular, $\Psi \circ \gamma$ is a solution to equation \eqref{eqn:LagFieldEqSect}, that is,
$$
(\Psi \circ \gamma)^*\inn(X)\Omega_\Lag = 0 \, , \ \mbox{for every } X \in \vf(J^1\pi) \, .
$$
Calculating, we have
$$
(\Psi \circ \gamma)^*\inn(X)\Omega_\Lag = \gamma^*(\Psi^*\inn(X)\Omega_\Lag) =
\gamma^*\inn(Y)\Psi^*\Omega_\Lag \, ,
$$
where $Y \in \vf(E)$ is a vector field $\Psi$-related with $X$. Nevertheless, since equation
\eqref{eqn:LagFieldEqSect} holds for every vector field in $J^1\pi$, we have proved
$$
\gamma^*\inn(Y)(\Psi^*\Omega_\Lag) = 0 \, , \ \mbox{for every } Y \in \vf(E) \, .
$$

For the converse, let $\gamma \in \Gamma(\pi)$ be an integral section of the distribution
$\Tan\pi^1(\restric{\D_\Lag}{\Im(\Psi)})$, and $\Psi \in \Gamma(\pi^1)$ an integrable jet field.
By the hypothesis we have that $\gamma^*\inn(Y)(\Psi^*\Omega_\Lag) = 0$ for every $Y \in \vf(E)$.
Then, computing, we have
$$
\gamma^*\inn(Y)(\Psi^*\Omega_\Lag) = \gamma^*(\Psi^*\inn(X)\Omega_\Lag =
(\Psi \circ \gamma)^*\inn(X)\Omega_\Lag \, ,
$$
where $X \in \vf(J^1\pi)$ is a vector field $\Psi$-related to $Y$. In particular, $X$ is tangent to
$\Im(\Psi) \hookrightarrow J^1\pi$ by construction, and we have proved that
$\Psi \circ \gamma \in \Gamma(\bar{\pi}^1)$ is a solution to equation
$$
(\Psi \circ \gamma)^*\inn(X)\Omega_\Lag = 0 \, ,
\ \mbox{for every } X \in \vf(J^1\pi) \mbox{ tangent to } \Im(\Psi) \, ,
$$
that is, to equation \eqref{eqn:LagFieldEqSectWithHJSolution}, which is equivalent to equation
\eqref{eqn:LagFieldEqSect} for a section of the form $\Psi \circ \gamma$. Therefore, the section
$\Psi \circ \gamma \in \Gamma(\bar{\pi}^1)$ is an integral section of $\D_\Lag$, and therefore
$\Psi$ is a solution to the generalized Lagrangian Hamilton-Jacobi problem.
\end{proof}

A straightforward consequence of the above result is the following.

\begin{corollary}\label{corol:LagGenHJProjectedSections}
Let $\Psi \in \Gamma(\pi^1)$ be an integrable jet field solution to the generalized Lagrangian
Hamilton-Jacobi problem. Then the integral sections of $\D_\Lag$ with boundary conditions in
$\Im(\Psi)$ project to the integral sections of $\D = \Tan\pi^1(\restric{\D_\Lag}{\Im(\Psi)})$.
\end{corollary}
\begin{proof}
Let $\Psi \in \Gamma(\pi^1)$ be a solution to the generalized Lagrangian Hamilton-Jacobi problem,
and $\psi_\Lag \in \Gamma(\bar{\pi}^1)$ an integral section of $\D_\Lag$ with boundary conditions
in $\Im(\Psi)$. Then, since the distribution $\D_\Lag$ is tangent to $\Im(\Psi)$ by Proposition
\ref{prop:LagGenHJEquivalences}, we have that $\Im(\psi_\Lag) \subseteq \Im(\Psi)$, and hence
$$
\Tan_{u}\Im(\pi^1 \circ \psi_\Lag) = \Tan_{\Psi(u)}\pi^1 (\Tan_{\Psi(u)}\Im(\psi_\Lag))
= \Tan_{\Psi(u)}\pi^1((\D_\Lag)_{\Psi(u)}) = \D_u \, ,
$$
where we have used that $\Im(\psi_\Lag) \subseteq \Im(\Psi)$ and
$\Psi \circ \pi^1 = \restric{\Id}{\Im(\Psi)}$.
\end{proof}

\noindent\textbf{Coordinate expression.}
Let $(x^i)$, $1 \leqslant i \leqslant m$, be local coordinates in $M$ such that
$\eta = \d^mx = \d x^1 \wedge \ldots \wedge \d x^m$, and $(x^i,u^\alpha)$,
$1 \leqslant \alpha \leqslant n$ local coordinates in $E$ adapted to the bundle structure. Then,
the induced coordinates in $J^1\pi$ are $(x^i,u^\alpha,u_i^\alpha)$, which coincide with the
local coordinates adapted to the bundle structure $\pi^1 \colon J^1\pi \to E$. In these coordinates,
a jet field $\Psi \in \Gamma(\pi^1)$ is given locally by $\Psi(x^i,u^\alpha) = (x^i,u^\alpha,\psi^\alpha_i)$,
where $\psi^\alpha_i(x^i,u^\alpha)$ are local smooth functions on $E$.

Let us compute the local condition for a jet field $\Psi \in \Gamma(\pi^1)$ to be a solution to the
generalized Lagrangian Hamilton-Jacobi problem. From Proposition \ref{prop:LagGenHJEquivalences}
we know that this is equivalent to require the distribution $\D_\Lag$ in $J^1\pi$ to be tangent to
the submanifold $\Im(\Psi) \hookrightarrow J^1\pi$, or, in terms of the class of multivector fields
$\{ \X_\Lag \} \subseteq \vf^{m}(J^1\pi)$ associated to $\D_\Lag$, to require every multivector
field in the class to be tangent to $\Im(\Psi)$.
From \cite{art:Echeverria_Munoz_Roman98} we know that a representative $\X_\Lag \in \{ \X_\Lag \}$
which is locally decomposable, $\bar{\pi}^1$-transverse and semi-holonomic may be chosen to have the
following coordinate expression
$$
\X_\Lag = \bigwedge_{j=1}^{m} \left( \derpar{}{x^j} + u_j^\alpha\derpar{}{u^\alpha}
+ F_{j,i}^\alpha \derpar{}{u_i^\alpha} \right) \, ,
$$
where the functions $F_{j,i}^\alpha$ are the solutions to the Euler-Lagrange equations
\begin{equation}\label{eqn:LagFieldEqLocal}
\derpar{L}{u^\alpha} - \derpars{L}{u_i^\alpha}{x^i} - u_i^\beta\derpars{L}{u_i^\alpha}{u^\beta}
-F_{j,i}^\beta\derpars{L}{u_i^\alpha}{u_j^\beta} = 0 \, ,
\end{equation}
in addition to the integrability conditions \eqref{eqn:Integrability&SemiHolonomicity} (if necessary).
Observe that every multivector field in the class $\{ \X_\Lag \}$ is obtained by multiplying this
representative by an arbitrary non-vanishing function $f \in \Cinfty(J^1\pi)$. Then, bearing in mind
that the submanifold $\Im(\Psi) \hookrightarrow J^1\pi$ is locally defined by the $mn$ constraints
$\psi_k^\beta - u_k^\beta = 0$, the condition for this particular $\X_\Lag$ to be tangent to
$\Im(\Psi)$ gives the following partial differential equations
\begin{equation}\label{eqn:LagGenHJLocal}
\derpar{\psi_k^\beta}{x^j} + u_j^\alpha\derpar{\psi_k^\beta}{u^\alpha} - \restric{F_{j,k}^\beta}{\Im(\Psi)} = 0 \, .
\end{equation}
This is a system of $nm^2$ partial differential equations with $nm$ unknown functions $\psi_k^\beta$,
that is, we have more equations than unknown functions.

\begin{remark}
Recall that the $n$ equations \eqref{eqn:LagFieldEqLocal} do not enable us to determinate all the
$m^2n$ coefficient functions $F_{i,j}^\alpha$, and, in general, there are $n(m^2-1)$ arbitrary
functions. Therefore, equations \eqref{eqn:LagGenHJLocal} may fix not only the coefficients
$\psi^\beta_k$ of the jet field $\Psi$, but also some of the remaining functions $F_{i,j}^\alpha$
of the Euler-Lagrange multivector fields which are solutions to the field equation
\eqref{eqn:LagFieldEqMVF}. In this way, we have a system of $m^2n$ partial differential
equations with $n(m^2 + m -1)$ unknown functions. Note that, even in the most favorable cases, there
still are $n(m-1)$ arbitrary functions to be determined, which may be fixed by the integrability
condition \eqref{eqn:Integrability&SemiHolonomicity} or not.
\end{remark}

\begin{remark}
On time-dependent mechanics, that is, for $m=1$ we obtain exactly $n$ partial differential equations
and $n$ unknown functions, since there are no arbitrary functions on the vector field solution
to the Lagrangian dynamical equation.
\end{remark}

\subsection{The Lagrangian Hamilton-Jacobi problem}

As in mechanics (see
\cite{art:Carinena_Gracia_Marmo_Martinez_Munoz_Roman06,art:Colombo_DeLeon_Prieto_Roman14_JPA}),
to solve the generalized Lagrangian Hamilton-Jacobi problem is, in general, a very difficult
task, since it amounts to find $(m+n)$-dimensional submanifolds of $J^1\pi$ such that the
$m$-dimensional distribution $\D_\Lag$ is tangent to them. Because of this, we impose an additional
condition on the jet field $\Psi \in \Gamma(\pi^1)$ in order to consider a less general problem.

\begin{definition}\label{def:LagHJDef}
The \textnormal{Lagrangian Hamilton-Jacobi problem} consists in finding a jet field
$\Psi \in \Gamma(\pi^1)$ solution to the generalized Lagrangian Hamilton-Jacobi problem
satisfying that $\Psi^*\Omega_\Lag = 0$. Such a jet field is called a \textnormal{solution to
the Lagrangian Hamilton-Jacobi problem}.
\end{definition}

With this new assumption we can state the following result, which is a straightforward
consequence of Proposition \ref{prop:LagGenHJEquivalences}, Corollary
\ref{corol:LagGenHJProjectedSections} and the results in \cite{art:Cantrijn_Ibort_DeLeon99}.

\begin{proposition}\label{prop:LagHJEquivalences}
Let $\Psi \in \Gamma(\pi^1)$ be an integrable jet field satisfying $\Psi^*\Omega_\Lag = 0$.
Then the following statements are equivalent:
\begin{enumerate}
\item $\Psi$ is a solution to the Lagrangian Hamilton-Jacobi problem.
\item The submanifold $\Im(\Psi) \hookrightarrow J^1\pi$ is $m$-Lagrangian and the distribution
$\D_\Lag$ is tangent to it.
\item The integral sections of $\D_\Lag$ with boundary conditions in $\Im(\Psi)$ project onto
the integral sections of $\D = \Tan\pi^1(\restric{\D_\Lag}{\Im(\Psi)})$.
\end{enumerate}
\end{proposition}

\noindent\textbf{Coordinate expression.}
In coordinates, we have
\begin{align*}
\Psi^*\Omega_\Lag &= \left( \derpars{L}{u_i^\alpha}{u^\beta} + \derpars{L}{u_i^\alpha}{u_k^\delta}\derpar{\psi_k^\delta}{u^\beta}
\right) \d u^\alpha \wedge \d u^\beta \wedge \d^{m-1}x_i \\
&\quad{} + \left( \derpars{L}{u_i^\alpha}{u_k^\beta} \derpar{\psi_k^\beta}{x^i} + \psi_i^\beta \derpars{L}{u_i^\beta}{u^\alpha}
- \derpar{L}{u^\alpha} + \psi_i^\delta\derpars{L}{u_i^\delta}{u_k^\beta}\derpar{\psi_k^\beta}{u^\alpha} \right) \d u^\alpha \wedge \d^mx \, .
\end{align*}
Hence, the condition $\Psi^*\Omega_\Lag = 0$ gives the following system of $n(1 + m(n-1))$ partial differential equations
$$
\derpars{L}{u_i^\alpha}{u^\beta} + \derpars{L}{u_i^\alpha}{u_k^\delta}\derpar{\psi_k^\delta}{u^\beta} = 0 \quad ; \quad
\derpars{L}{u_i^\alpha}{u_k^\beta} \derpar{\psi_k^\beta}{x^i} + \psi_i^\beta \derpars{L}{u_i^\beta}{u^\alpha}
- \derpar{L}{u^\alpha} + \psi_i^\delta\derpars{L}{u_i^\delta}{u_k^\beta}\derpar{\psi_k^\beta}{u^\alpha} = 0 \, ,
$$
where $1 \leqslant i \leqslant m$, $1 \leqslant \alpha < \beta \leqslant n$ in the first set, and
$1 \leqslant \alpha \leqslant n$ in the second. These two sets of equations may be combined to
obtain the following system 
\begin{equation}\label{eqn:LagHJLocal}
\derpars{L}{u_i^\alpha}{u^\beta} + \derpars{L}{u_i^\alpha}{u_k^\delta}\derpar{\psi_k^\delta}{u^\beta} = 0 \quad ; \quad
\derpars{L}{u_i^\alpha}{u_k^\beta} \derpar{\psi_k^\beta}{x^i} - \derpar{L}{u^\alpha} = 0 \, .
\end{equation}

On the other hand, observe that since $\Psi^*\Omega_\Lag = -\d(\Psi^*\Theta_\Lag)$, the condition
$\Psi^*\Omega_\Lag = 0$ is equivalent to requiring the $m$-form $\Psi^*\Theta_\Lag \in \df^{m}(E)$
to be closed. In particular, using Poincar\'{e}'s Lemma, the $m$-form $\Psi^*\Theta_\Lag$ is
locally exact, that is, there exists a $(m-1)$-form $\omega \in \df^{m-1}(U)$, with $U \subseteq E$
an open set, such that $\d\omega = \Psi^*\Theta_\Lag$. Moreover, since $\Theta_\Lag$ is
$\pi^1$-semibasic, so is $\Psi^*\Theta_\Lag$, and therefore $\omega$ must be $\pi$-semibasic.
In coordinates, bearing in mind the coordinate expression \eqref{eqn:CartanMFormLocal} of the
Cartan $m$-form, we obtain
$$
\Psi^*\Theta_\Lag = \derpar{L}{u_i^\alpha} \d u^\alpha \wedge \d^{m-1}x_i
- \left( \derpar{L}{u_i^\alpha} \psi_i^\alpha - L \right) \d^mx \, .
$$
In addition, the coordinate expression for a generic $\pi$-semibasic local $(m-1)$-form $\omega$ in $E$ is
$$
\omega = W^i \d^{m-1}x_i \, ,
$$
where $W^i \in \Cinfty(E)$ are local functions. From this we deduce the local expression of the
$m$-form $\d\omega$, which is
$$
\d\omega = \sum_{i=1}^{m}\derpar{W^i}{x^i} \, \d^{m}x
+ \derpar{W^i}{u^\alpha} \, \d u^\alpha \wedge \d^{m-1}x_i \, .
$$
Finally, requiring $\d\omega = \Psi^*\Theta_\Lag$, we obtain
$$
\sum_{i=1}^{m}\derpar{W^i}{x^i} + \psi_i^\alpha\restric{\derpar{L}{u_i^\alpha}}{\Im(\Psi)} - L(x^i,u^\alpha,\psi^\alpha_i) = 0 \quad ; \quad
\derpar{W^i}{u^\alpha} = \restric{\derpar{L}{u_i^\alpha}}{\Im(\Psi)} \, ,
$$
which may be combined to give the equation
\begin{equation}\label{eqn:LagHJEq}
\sum_{i=1}^{m}\derpar{W^i}{x^i} + \psi_i^\alpha\derpar{W^i}{u^\alpha} - L(x^i,u^\alpha,\psi^\alpha_i) = 0 \, ,
\end{equation}
which is the Hamilton-Jacobi equation in the Lagrangian formalism.

\subsection{Complete solutions}
\label{sec:LagHJCompleteSol}

In the above Sections we stated the Hamilton-Jacobi problem in the Lagrangian formalism, and a
jet field $\Psi \in \Gamma(\pi^1)$ solution to this problem gives a particular solution to the
Lagrangian problem in the form of a submanifold the phase space $J^1\pi$. Nevertheless, this
is not a complete solution to the Lagrangian problem, since only the integral sections of the
distribution $\D_\Lag$ with boundary conditions in $\Im(\Psi)$ can be recovered from the solution
to the Hamilton-Jacobi problem.

Hence, in order to obtain a complete solution to the problem, we need to endow the phase space
$J^1\pi$ with a foliation such that every leaf is the image set of a jet field solution to the
Lagrangian Hamilton-Jacobi problem. The precise definition is:
 
\begin{definition}
A \textnormal{complete solution to the Lagrangian Hamilton-Jacobi problem} is a local diffeomorphism
$\Phi \colon U \times E \to J^1\pi$, with $U \subseteq \R^{mn}$ an open set, such that for every
$\lambda \in U$, the map $\Psi_\lambda(\bullet) \equiv \Phi(\lambda,\bullet) \colon E \to J^1\pi$
is a jet field in $E$ solution to the Lagrangian Hamilton-Jacobi problem.
\end{definition}

\begin{remark}
An alternative, but equivalent, definition of a complete solution consists in giving the full set
of jet fields $\{ \Psi_\lambda \in \Gamma(\pi^1) \mid \lambda \in U \subseteq \R^{mn} \}$ depending
on $mn$ parameters, instead of the local diffeomorphism $\Phi$.
\end{remark}

From the definition we deduce that a complete solution endows the Lagrangian phase space $J^1\pi$
with a foliation transverse to the fibers such that every leaf has dimension $m+n$ and the
distribution $\D_\Lag$ is tangent to it.

It follows from this last comment that a complete solution to the Lagrangian Hamilton-Jacobi problem
enables us to recover every integral section of the distribution $\D_\Lag$ solution to the Lagrangian
problem, that is, we can recover every section solution to the Euler-Lagrange equations for classical
field theories. In particular, let $\Phi$ be a complete solution, and let us consider the following
set of distributions in $E$:
$$
\left\{ \D_\lambda = \Tan\pi^1(\restric{\D_\Lag}{\Im(\Psi_\lambda)}) \subseteq \Tan E \mid
\lambda \in U \subseteq \R^{mn} \right\} \, ,
$$
where $\Psi_\lambda(\bullet) \equiv \Phi(\lambda,\bullet)$. Then, the integral sections of $\D_\lambda$,
for different values of the parameter $\lambda \in U$, provide all the integral sections of the
distribution $\D_\Lag$ solution to the Lagrangian problem. Indeed, let $j^1_x\phi \in J^1\pi$ be a
point, and let us denote $u = \phi(x) = \pi^1(j^1_x\phi)$. Then, since $\Phi$ is a complete solution,
there exists $\lambda_o \in U$ such that $\Phi(\lambda_o,u) \equiv \Psi_{\lambda_o}(u) = j_x^1\phi$,
and the integral sections of $\D_{\lambda_o}$ through $u$, composed with $\Psi_{\lambda_o}$, give the
integral sections of $\D_\Lag$ through $j_x^1\phi$.

\section{The Hamilton-Jacobi problem in the Hamiltonian formalism}
\label{sec:HamiltonianFormalism}

As in the Lagrangian formalism, the configuration bundle in the Hamiltonian formulation for
multisymplectic classical field theories is a bundle $\pi \colon E \to M$, where $M$ is a
$m$-dimensional orientable manifold with fixed volume form $\eta \in \df^{m}(M)$, and $\dim E = m + n$.
Two phase spaces are considered in this formulation: the extended multimomentum bundle
$\Lambda^m_2(\Tan^*E)$ and the restricted multimomentum bundle $J^1\pi^*$ introduced in Section
\ref{geomset}.

The physical information is given in terms of a \textsl{Hamiltonian section} $h \in \Gamma(\mu)$,
which is specified by a \textsl{local Hamiltonian function} $H \in \Cinfty(J^1\pi^*)$, that is,
we have $h(x^i,u^\alpha,p_\alpha^i) = (x^i,u^\alpha,-H,p_\alpha^i)$. Then, from the canonical
forms in $\Lambda^m_2(\Tan^*E)$, and using this Hamiltonian section, the Hamilton-Cartan forms
$\Theta_h = h^*\Theta \in \df^{m}(J^1\pi^*)$ and
$\Omega_h = h^*\Omega = -\d\Theta_h \in \df^{m+1}(J^1\pi^*)$ are constructed, with coordinate expressions
$$
\Theta_h = p_\alpha^i\d u^\alpha \wedge \d^{m-1}x_i - H \d^mx \quad ; \quad
\Omega_h = -\d p_\alpha^i \wedge \d u^\alpha \wedge \d^{m-1}x_i + \d H \wedge \d^{m}x \, .
$$

Then, the Hamiltonian problem for a first-order classical field theory is the following: \textit{to
find $m$-dimensional, $\bar{\pi}_E^r$-transverse and integrable distribution $\D_h$ in $J^1\pi^*$ such
that the integral sections $\psi_h \in \Gamma(\bar{\pi}_E^r)$ of $\D_h$ are solutions to the field
equation}
\begin{equation}\label{eqn:HamFieldEqSect}
\psi_h^*\inn(X)\Omega_h = 0 \, , \ \mbox{\textit{for every} } X \in \vf(J^1\pi^*) \, .
\end{equation}
Contrary to the Lagrangian formalism for classical field theories, the $(m+1)$-form
$\Omega_h \in \df^{m+1}(J^1\pi^*)$ is multisymplectic regardless of the Hamiltonian section
$h \in \Gamma(\mu)$ provided, in the same way that it occurs in Classical Mechanics in the Lagrangian
and Hamiltonian formalisms. Therefore, there exists such a distribution $\D_h$, although it is
not necessarily integrable. In the following we assume that the distribution $\D_h$ is integrable.

From the results in Section \ref{sec:MultivectorFields}, this distribution $\D_h$ is associated
with a class of integrable and $\bar{\pi}_E^r$-transverse multivector fields
$\{ \X_h \} \subseteq \vf^{m}(J^1\pi^*)$ satisfying
\begin{equation}\label{eqn:HamFieldEqMVF}
\inn(\X_h)\Omega_h = 0 \, , \ \mbox{for every } \X_h \in \left\{ \X_h \right\} \, .
\end{equation}
Same comments apply: since the $(m+1)$-form $\Omega_h$ is $1$-nondegenerate a
$\bar{\pi}_E^r$-transverse solution to equation \eqref{eqn:HamFieldEqMVF} does exist, but it may
not be integrable. In the following we assume that every multivector field in the class $\{ \X_h \}$
is integrable.

(For more details on the Hamiltonian formalism of field theories see, for instance,
\cite{art:Carinena_Crampin_Ibort91,EMR-99b,GIMMSY-mm,Krva-2002,LMM-96,art:Roman09}).

\subsection{The generalized Hamiltonian Hamilton-Jacobi problem}
\label{sec:HamGenHJProblem}

Following the patterns in \cite{art:Carinena_Gracia_Marmo_Martinez_Munoz_Roman06} and in previous
Sections, we first state a generalized version of the Hamilton-Jacobi problem in the Hamiltonian
formalism.

\begin{definition}
The \textnormal{generalized Hamiltonian Hamilton-Jacobi problem} consists in finding a section
$s \in \Gamma(\pi_E^r)$ and a $m$-dimensional and integrable distribution $\D$ in $E$ such that
if $\gamma \in \Gamma(\pi)$ is an integral section of $\D$, then
$s \circ \gamma \in \Gamma(\bar{\pi}_E^r)$ is an integral section of $\D_h$, that is,
\begin{equation}\label{eqn:HamGenHJDefDist}
\Tan_{u}\Im(\gamma) = \D_{u} \ \ \forall u \in \Im(\gamma) \Longrightarrow
\Tan_{[\omega]}\Im(s \circ \gamma) = (\D_h)_{[\omega]}
\ \ \forall \ [\omega] \in \Im(s \circ \gamma) \, .
\end{equation}
\end{definition}

As in the Lagrangian formulation stated in Section \ref{sec:LagrangianFormalism}, from the results
in Section \ref{sec:MultivectorFields} we know that both $\D$ and $\D_h$ are associated with their
corresponding classes of multivector fields. Thus, we can state the generalized Hamiltonian
Hamilton-Jacobi problem in an equivalent way in terms of multivector fields as the search of a section
$s \in \Gamma(\pi_E^r)$ and a class of locally decomposable and integrable multivector fields
$\{ \X \} \subseteq \vf^{m}(E)$ such that if $\gamma \in \Gamma(\pi)$ is an integral section of
every multivector field $\X\in\{\X\}$, then $s \circ \gamma \in \Gamma(\bar{\pi}_E^r)$ is an integral
section of every multivector field $\X_h \in \{ \X_h \}$ solution to equation \eqref{eqn:HamFieldEqMVF},
that is,
\begin{equation}\label{eqn:HamGenHJDefMVF}
\X \circ \gamma = \Lambda^m\dot{\gamma} \ \ \forall \X \in \{\X\} \Longrightarrow
\X_h \circ s \circ \gamma = \Lambda^m(\dot{\overline{s \circ \gamma}})
\ \ \forall \X_h \in \{ \X_h \} \, .
\end{equation}
Again, as in the Lagrangian formalism we have the following diagram illustrating this equivalent
formulation of the generalized Hamiltonian Hamilton-Jacobi problem
$$
\xymatrix{
\ & \ & \Lambda^m(\Tan J^1\pi^*) & \ \\
\ & \ & J^1\pi^* \ar[u]_-{\X_h} \ar[d]^{\pi^1} & \ \\
M \ar[rr]^-{\gamma} \ar[urr]^-{s \circ \gamma} \ar@/^1pc/[uurr]^-{\Lambda^m(\dot{\overline{s \circ \gamma}})} \ar@/_1pc/[rrr]_{\Lambda^m\dot{\gamma}} & \ & E \ar@/^.75pc/[u]^{s} \ar[r]^-{\X} & \Lambda^m(\Tan E)
}
$$
where the interpretation is the same as in the corresponding diagram for the Lagrangian formalism:
\textit{if the lower diagram formed by $\gamma$, $\X$ and $\Lambda^m\dot{\gamma}$ is commutative
for every $\X \in \{ \X \}$, then the upper diagram formed by $s \circ \gamma$, $\X_h$ and
$\Lambda^m(\dot{\overline{s \circ \gamma}})$ is also commutative for every $\X_h \in \{ \X_h \}$.}

\begin{remark}
Analogously to the Lagrangian formulation, the section $s \circ \gamma \in \Gamma(\bar{\pi}_E^r)$
is an integral section of the distribution $\D_h$ solution to the Hamiltonian problem and, therefore,
it is a solution to the equation \eqref{eqn:HamFieldEqSect}, that is,
$$
(s \circ \gamma)^*\inn(X)\Omega_h = 0 \, , \ \mbox{for every } X \in \vf(J^1\pi^*) \, .
$$
Then, bearing in mind that the action of the $m$-form $(s \circ \gamma)^*\inn(X)\Omega_h \in \df^{m}(M)$
on $m$ tangent vectors $v_i \in \Tan_xM$ ($x \in M$) is defined as
$$
\left((s \circ \gamma)^*\inn(X)\Omega_h\right)_{x}(v_1,\ldots,v_m) =
(\Omega_h)_{s(\gamma(x))}(X(s(\gamma(x))),\Tan_x(s \circ \gamma)(v_1),\ldots,\Tan_x(s \circ \gamma)(v_m)) \, ,
$$
from where we observe that $X(s(\gamma(x)))\in\Tan_{s(\gamma(x))}\Im(s\circ\gamma)\subset\Tan_{s(\gamma(x))}\Im(s)$,
that is, the vector field $X$ is tangent to the submanifold $\Im(s) \hookrightarrow J^1\pi^*$.
Hence, we conclude that in this case equation \eqref{eqn:HamFieldEqSect} is equivalent to
\begin{equation*}
(s \circ \gamma)^*\inn(X)\Omega_h = 0 \, ,
\ \mbox{for every } X \in \vf(J^1\pi^*) \mbox{ tangent to } \Im(s) \, .
\end{equation*}
\end{remark}

It is clear from the Definition that the distribution $\D$ in $E$, the section $s \in \Gamma(\pi_E^r)$
and the Hamiltonian distribution $\D_h$ in $J^1\pi^*$ are closely related. In fact, we have the
following result, which is the analogous to Proposition \ref{prop:LagGenHJRelatedDist} in the
Hamiltonian formalism.

\begin{proposition}\label{prop:HamGenHJRelatedDist}
The section $s \in \Gamma(\pi_E^r)$ and the distribution $\D$ in $E$ satisfy condition
\eqref{eqn:HamGenHJDefDist} if, and only if, $\D$ and $\D_h$ are $s$-related.
\end{proposition}
\begin{proof}
This proof follows exactly the same patterns as the proof of Proposition
\ref{prop:HamGenHJRelatedDist}.
\end{proof}

As in the Lagrangian formalism, a straightforward consequence of this last result is the following.

\begin{corollary}\label{corol:HamGenHJAssociatedMVF}
If the section $s \in \Gamma(\pi_E^r)$ and the distribution $\D$ in $E$ satisfy condition
\eqref{eqn:HamGenHJDefDist}, then for every $\X \in \{\X\}$ there exists a multivector field
$\X_h \in \{\X_h\}$ such that $\X$ is given by
$$
\X = \Lambda^m\Tan\pi_E^r \circ \X_h \circ s \, .
$$
\end{corollary}
\begin{proof}
The proof follows the same patterns as in the proof of Corollary \ref{corol:LagGenHJAssociatedMVF}.
\end{proof}

That is, the class of integrable multivector fields $\{ \X \} \subseteq \vf^{m}(E)$ is completely
determinated by the section $s \in \Gamma(\pi_E^r)$ and the class of multivector fields
$\{ \X_h \} \subseteq \vf^{m}(J^1\pi^*)$ solution to the equation \eqref{eqn:HamFieldEqMVF},
and every multivector field $\X = \Lambda^m\Tan\pi^r_E \circ \X_h \circ s \in \{ \X \}$ is called
the \textsl{multivector field associated to $s$ and $\X_h$}. The diagram which illustrates this
situation is the following:
$$
\xymatrix{
J^1\pi^* \ar[rr]^-{\X_h} \ar[dd]^{\pi_E^r} & \ & \Lambda^m(\Tan J^1\pi^*) \ar[dd]_{\Lambda^m\Tan\pi_E^r} \\
\ & \ & \ \\
E \ar@/^1pc/[uu]^{s} \ar[rr]_-{\X} & \ & \Lambda^m(\Tan E) \ar@/_1pc/[uu]_{\Lambda^m\Tan s}
}
$$
As a consequence, the integrable distribution $\D$ in $E$ is completely determinated by the section
$s \in \Gamma(\pi_E^r)$ and the distribution $\D_h$ solution to the Hamiltonian problem as
$\D_{\pi_E^r([\omega])} = \Tan_{[\omega]}\pi_E^r((\D_h)_{[\omega]})$ for every $[\omega] \in J^1\pi^*$,
or $\D = \Tan\pi_E^r(\restric{\D_h}{\Im(s)})$, and it is called the \textsl{distribution associated
to $s$}.

From Corollary \ref{corol:HamGenHJAssociatedMVF} we deduce that the search for a section
$s \in \Gamma(\pi_E^r)$ and a distribution $\D$ in $E$ satisfying condition \eqref{eqn:HamGenHJDefDist}
is equivalent to the search of a section $s \in \Gamma(\pi^r_E)$ such that condition
\eqref{eqn:HamGenHJDefDist} is satisfied with the associated distribution
$\Tan\pi_E^r(\restric{\D_h}{\Im(s)})$. Therefore, we can give the following definition.

\begin{definition}
A \textnormal{solution to the generalized Hamiltonian Hamilton-Jacobi problem} is a section
$s \in \Gamma(\pi_E^r)$ such that if $\gamma \in \Gamma(\pi)$ is an integral section of the
$m$-dimensional and integrable distribution $\Tan\pi_E^r(\restric{\D_h}{\Im(s)})$ in $E$, then
$s \circ \gamma \in \Gamma(\bar{\pi}_E^r)$ is an integral section of $\D_h$.
\end{definition}

\begin{proposition}\label{prop:HamGenHJEquivalences}
Let $s \in \Gamma(\pi_E^r)$ be a section. Then, the following conditions are equivalent.
\begin{enumerate}
\item $s$ is a solution to the generalized Hamiltonian Hamilton-Jacobi problem.
\item The distribution $\D_h$ in $J^1\pi^*$ is tangent to the submanifold
$\Im(s) \hookrightarrow J^1\pi^*$, that is,
$(\D_h)_{[\omega]} \subseteq \Tan_{[\omega]}\Im(s)$ for every $[\omega] \in \Im(s)$.
\item $s$ satisfies the equation
$$
\gamma^*\inn(Y)\d(h \circ s) = 0 \, , \ \mbox{for every } Y \in \vf(E) \, ,
$$
where $\gamma \in \Gamma(\pi)$ is an integral section of the associated distribution
$\Tan\pi_E^r(\restric{\D_h}{\Im(s)})$.
\end{enumerate}
\end{proposition}
\begin{proof}
This proof follows the same patterns as the proof of Proposition \ref{prop:LagGenHJEquivalences},
bearing in mind the properties of the tautological $m$-form $\Theta \in \df^{m}(\Lambda_2^m(\Tan^*E))$,
that is, we have $\omega^*\Theta = \omega$ for every $\omega \in \df^{m}(E)$. Because of this, we have
\begin{equation}\label{eqn:PullBackMultisymplecticFormBySection}
s^*\Omega_h = s^*(h^*\Omega) = (h \circ s)^*\Omega = (h \circ s)^*(-\d\Theta) = -\d(h \circ s)^*\Theta
= -\d(h \circ s) \, ,
\end{equation}
and therefore the equation
$$
(s \circ \gamma)^*\inn(X)\Omega_h = 0 \, , \ \mbox{for every } X \in \vf(J^1\pi^*) \, ,
$$
gives rise to equation
$$
\gamma^*\inn(Y)\d(h \circ s) = 0 \, , \ \mbox{for every } Y \in \vf(E) \, .
$$
\end{proof}

As in the Lagrangian formalism, a consequence of Proposition \ref{prop:HamGenHJEquivalences}
is the following result.

\begin{corollary}\label{corol:HamGenHJProjectedSections}
Let $s \in \Gamma(\pi_E^r)$ be a section solution to the generalized Hamiltonian
Hamilton-Jacobi problem. Then the integral sections of $\D_h$ with boundary conditions in $\Im(s)$
project to the integral sections of $\D = \Tan\pi_E^r(\restric{\D_h}{\Im(s)})$.
\end{corollary}
\begin{proof}
This proof is analogous to the proof of Corollary \ref{corol:LagGenHJProjectedSections}.
\end{proof}

\noindent\textbf{Coordinate expression.}
Let $(x^i)$, $1 \leqslant i \leqslant m$, be local coordinates in $M$ such that
$\eta = \d^mx = \d x^1 \wedge \ldots \wedge \d x^m$, and $(x^i,u^\alpha)$,
$1 \leqslant \alpha \leqslant n$ local coordinates in $E$ adapted to the bundle structure. Then,
the induced coordinates in $J^1\pi^*$ are $(x^i,u^\alpha,p_\alpha^i)$, which coincide with the
local coordinates adapted to the bundle structure $\pi_E^r \colon J^1\pi^* \to E$. In these coordinates,
a section $s \in \Gamma(\pi_E^r)$ is given locally by $s(x^i,u^\alpha) = (x^i,u^\alpha,s_\alpha^i)$,
where $s_\alpha^i(x^i,u^\alpha)$ are local smooth functions on $E$.

Let us compute the local condition for a section $s \in \Gamma(\pi_E^r)$ to be a solution to the
generalized Hamiltonian Hamilton-Jacobi problem. From Proposition \ref{prop:HamGenHJEquivalences}
we know that this is equivalent to require the distribution $\D_h$ solution to the Hamiltonian
problem to be tangent to the submanifold $\Im(s) \hookrightarrow J^1\pi^*$, or, in terms of the
associated class of multivector fields $\{ \X_h \} \subseteq \vf^{m}(J^1\pi^*)$, to require
every multivector field in the class to be tangent to $\Im(s)$. From \cite{art:Roman09} we know
that a representative $\X_h \in \{ \X_h \}$ which is locally decomposable and
$\bar{\pi}_E^r$-transverse may be chosen to have the following coordinate expression
$$
\X_h = \bigwedge_{j=1}^{m} \left( \derpar{}{x^j} + \derpar{H}{p_\alpha^j} \derpar{}{u^\alpha}
+ G_{\alpha,j}^i \derpar{}{p_\alpha^i} \right) \, ,
$$
where the functions $G_{\alpha,j}^i$ are the solutions to the equations
\begin{equation}\label{eqn:HamFieldEqLocal}
\sum_{i=1}^{m}G_{\alpha,i}^{i} = -\derpar{H}{u^\alpha} \, ,
\end{equation}
in addition to the integrability conditions (if necessary), which in this case give the following
system of $nm(m^2-1)/2$ partial differential equations for the component functions $G_{\alpha,j}^i$
\begin{equation}\label{eqn:Integrability}
\begin{array}{l}
\displaystyle
\derpars{H}{x^j}{p_\alpha^k} + \derpar{H}{p_\beta^j}\,\derpars{H}{u^\beta}{p_\alpha^k}
+ G_{\beta,j}^l\derpars{H}{p_\beta^l}{p_\alpha^k} - \derpars{H}{x^k}{p_\alpha^j}
- \derpar{H}{p_\beta^k}\,\derpars{H}{u^\beta}{p_\alpha^j} - G_{\beta,k}^l\derpars{H}{p_\beta^l}{p_\alpha^j} = 0 \, , \\[15pt]
\displaystyle
\derpar{G_{\alpha,k}^i}{x^j} + \derpar{H}{p_\beta^j}\,\derpar{G_{\alpha,k}^i}{u^\beta}
+ G_{\beta,j}^l\derpar{G_{\alpha,k}^i}{p_\beta^l} - \derpar{G_{\alpha,j}^i}{x^k}
- \derpar{H}{p_\beta^k}\,\derpar{G_{\alpha,j}^i}{u^\beta} - G_{\beta,k}^l\derpar{G_{\alpha,j}^i}{p_\beta^l} = 0 \, .
\end{array}
\end{equation}

Then, bearing in mind that the submanifold $\Im(s) \hookrightarrow J^1\pi^*$ is locally defined by
the $mn$ constraints $s_\beta^k - p_\beta^k = 0$, the condition for this representative of the class
to be tangent to $\Im(s)$ gives the following partial differential equations
\begin{equation}\label{eqn:HamGenHJLocal}
\restric{\derpar{s^k_\beta}{x^j} + \derpar{H}{p_\alpha^j}\derpar{s^k_\beta}{u^\alpha} - G_{\beta,j}^k}{\Im(s)} = 0 \, .
\end{equation}
This is a system of $nm^2$ partial differential equations with $nm$ unknown functions $s^k_\beta$,
that is, we have more equations than unknown functions.

\begin{remark}
Recall that equations \eqref{eqn:HamFieldEqLocal} do not determinate uniquely all the coefficient
functions $G_{\alpha,j}^i$, since there are $m^2n$ unknown functions $G_{\alpha,j}^i$ and we have
only $n$ equations, which implies that, in general, there are $n(m^2-1)$ arbitrary functions.
Therefore, equations \eqref{eqn:HamGenHJLocal} could enable us to fix some of the arbitrary functions
$G_{\alpha,j}^{i}$ of the Hamiltonian multivector fields solution to the field equation
\eqref{eqn:HamFieldEqMVF}. From this point of view, equations \eqref{eqn:HamGenHJLocal} are a system
of $m^2n$ partial differential equations with $n(m^2 + m - 1)$ unknown functions. Note that in the
most favorable cases there still are $n(m-1)$ arbitrary functions to be determined, which may be fixed
by the integrability condition \eqref{eqn:Integrability} or not.
\end{remark}

\begin{remark}
On time-dependent mechanics, that is, for $m=1$ we obtain exactly $n$ partial differential equations
and $n$ unknown functions, since there are no arbitrary functions on the vector field solution
to the Hamiltonian dynamical equation.
\end{remark}

\subsection{The Hamiltonian Hamilton-Jacobi problem}

As in mechanics (see
\cite{art:Carinena_Gracia_Marmo_Martinez_Munoz_Roman06,art:Colombo_DeLeon_Prieto_Roman14_JPA}),
to solve the generalized Hamiltonian Hamilton-Jacobi problem is, in general, a very difficult
task, since it amounts to find $mn$-codimensional submanifolds of $J^1\pi^*$ such that the
$m$-dimensional distribution $\D_h$ is tangent to them. For this reason we require an additional
condition to the section $s \in \Gamma(\pi_E^r)$, and thus we consider a more particular problem.

\begin{definition}\label{def:HamHJDef}
The \textnormal{Hamiltonian Hamilton-Jacobi problem} consists in finding a section
$s \in \Gamma(\pi_E^r)$ solution to the generalized Hamiltonian Hamilton-Jacobi problem
such that $s^*\Omega_h = 0$. Such a section is called a \textnormal{solution to the
Hamiltonian Hamilton-Jacobi problem}.
\end{definition}

\begin{remark}
Bearing in mind the properties of the tautological $m$-form $\Theta \in \df^{m}(\Lambda^m_2(\Tan^*E))$,
and the calculations in \eqref{eqn:PullBackMultisymplecticFormBySection}, the condition $s^*\Omega_h = 0$
is equivalent to the closedness of the $m$-form $h \circ s \in \df^{m}(E)$.
\end{remark}

With this new assumption we can state the following result, which is a straightforward
consequence of Proposition \ref{prop:HamGenHJEquivalences}, Corollary
\ref{corol:HamGenHJProjectedSections} and the results on isotropic submanifold of multisymplectic
manifolds in \cite{art:Cantrijn_Ibort_DeLeon99}.

\begin{proposition}\label{prop:HamHJEquivalences}
Let $s \in \Gamma(\pi_E^r)$ be a section satisfying $s^*\Omega_h = 0$. Then the
following statements are equivalent:
\begin{enumerate}
\item $s$ is a solution to the Hamiltonian Hamilton-Jacobi problem.
\item The submanifold $\Im(s) \hookrightarrow J^1\pi^*$ is $m$-Lagrangian and the distribution $\D_h$
solution to the Hamiltonian problem is tangent to it.
\item The integral sections of $\D_h$ with boundary conditions in $\Im(s)$ project onto the integral
sections of $\D = \Tan\pi_E^r(\restric{\D_h}{\Im(s)})$.
\end{enumerate}
\end{proposition}

\noindent\textbf{Coordinate expression.}
In coordinates, we have
$$
s^*\Omega_h = -\d(h \circ s) = \left( \derpar{H}{u^\alpha} + \derpar{H}{p_\beta^j}\derpar{s_\beta^j}{u^\alpha} +
\sum_{i=1}^{m} \derpar{s_\alpha^i}{x^i} \right) \d u^\alpha \wedge \d^m x
+ \derpar{s_\alpha^i}{u^\beta} \, \d u^\alpha \wedge \d u^\beta \wedge \d^{m-1}x_i \, .
$$
Hence, the condition $s^*\Omega_h = 0$ or, equivalently, $\d(h \circ s) = 0$, gives the following
$n(1 + m(n-1))$ partial differential equations
\begin{equation}\label{eqn:HamHJLocal}
\derpar{H}{u^\alpha} + \derpar{H}{p_\beta^j}\derpar{s_\beta^j}{u^\alpha} +
\sum_{i=1}^{m} \derpar{s_\alpha^i}{x^i} = 0 \quad ; \quad
\derpar{s_\alpha^i}{u^\beta} - \derpar{s_\beta^i}{u^\alpha} = 0 \, ,
\end{equation}
where $1 \leqslant \alpha \leqslant n$ in the first set, and $1 \leqslant \alpha < \beta \leqslant n$
in the second. Therefore, a section $s \in \Gamma(\pi_E^r)$ solution to the Hamiltonian
Hamilton-Jacobi problem must satisfy the $m^2n + mn(n-1) + n$ partial differential equations
\eqref{eqn:HamGenHJLocal} and \eqref{eqn:HamHJLocal}.

\begin{remark}
Since equations \eqref{eqn:HamHJLocal} are the analogous to equations \eqref{eqn:LagHJLocal} in the
Hamiltonian formalism we deduce that, in general, equations \eqref{eqn:HamHJLocal} may not be locally
linear independent among themselves or together with \eqref{eqn:HamGenHJLocal}. In particular, we may
have less equations than the $m^2n + mn(n-1) + n$ given by \eqref{eqn:HamGenHJLocal} and
\eqref{eqn:HamHJLocal}.
\end{remark}

Now we recover the classic Hamilton-Jacobi equation for first-order classical field theories. Since
the $\pi_E^r$-semibasic $m$-form $h \circ s$ is closed, by Poincar\'{e}'s Lemma it is locally exact,
that is, there exists a $\pi$-semibasic $(m-1)$-form $\omega \in \df^{m-1}(U)$, with $U \subseteq E$
an open set, such that
$\d\omega = h \circ s$. In coordinates, a semibasic $(m-1)$-form defined in an open set of $E$ is given by
$$
\omega = W^i \d^{m-1}x_i \, ,
$$
with $W^i \in \Cinfty(E)$ being local functions. From this, the $m$-form $\d\omega$ is given by
$$
\d\omega = \sum_{i=1}^{m}\derpar{W^i}{x^i} \, \d^{m}x
+ \derpar{W^i}{u^\alpha} \, \d u^\alpha \wedge \d^{m-1}x_i \, .
$$
Hence, requiring the identity $\d\omega = h \circ s$ to hold, we obtain
$$
-H(x^i,u^\alpha,s_\alpha^i) = \sum_{i=1}^{m} \derpar{W^i}{x^i} \quad ; \quad
\derpar{W^i}{u^\alpha} = s_\alpha^i \, ,
$$
from where the classic Hamilton-Jacobi equation is deduced
\begin{equation}\label{eqn:HamHJEq}
\sum_{i=1}^{m} \derpar{W^i}{x^i} + H\left(x^i,u^\alpha,\derpar{W^i}{u^\alpha}\right) = 0 \, .
\end{equation}

\subsection{Complete solutions}
\label{sec:HamHJCompleteSol}

As in the Lagrangian formalism stated in Section \ref{sec:LagrangianFormalism}, in the previous
Sections we stated the Hamiltonian formulation of the Hamilton-Jacobi problem for first-order
classical field theories. As we have proved previously, a section $s \in \Gamma(\pi_E^r)$ solution
to the Hamilton-Jacobi problem gives rise to a particular set of solutions to the Hamiltonian
problem in terms of a submanifold of phase space $J^1\pi^*$. Observe, however, that this section
of the bundle $J^1\pi^* \to E$ is not a complete solution to the Hamiltonian problem, since only
the integral sections of $\D_h$ with boundary conditions in $\Im(s)$ can be recovered from the
solution to the Hamilton-Jacobi problem.

Therefore, a complete solution to the problem is given in terms of a foliation of the phase space
$J^1\pi^*$ such that every leaf is the image set of a section solution to the Hamiltonian
Hamilton-Jacobi problem.

\begin{definition}
A \textnormal{complete solution to the Hamiltonian Hamilton-Jacobi problem} is a local diffeomorphism
$\Phi \colon U \times E \to J^1\pi^*$, with $U \subseteq \R^{mn}$ an open set, such that for every
$\lambda \in U$, the map $s_\lambda(\bullet) \equiv \Phi(\lambda,\bullet) \colon E \to J^1\pi^*$
is a section of the projection $\pi_E^r$ which is a solution to the Hamiltonian Hamilton-Jacobi problem.
\end{definition}

As in the Lagrangian formalism, it follows from the definition that a complete solution provides
the manifold $J^1\pi^*$ with a foliation transverse to the fibers satisfying that every leaf
is $(m+n)$-dimensional and the distribution $\D_h$ is tangent to it.

Moreover, from a complete solution to the Hamiltonian Hamilton-Jacobi problem we can recover
every section solution to the Hamilton-De Donder-Weyl equations. In particular, let
$\Phi$ be a complete solution, and let us consider the following set of distributions in $E$
$$
\left\{ \D_\lambda =  \Tan\pi_E^r(\restric{\D_h}{\Im(s_\lambda)}) \subset \Tan E \mid
\lambda \in U \subseteq \R^{mn} \right\} \, ,
$$
where $s_\lambda(\bullet) \equiv \Phi(\lambda,\bullet)$. Then, the integral sections of $\D_\lambda$,
for different values of $\lambda \in U$, provide all the integral sections of the distribution $\D_h$
solution to the Hamiltonian problem. Indeed, let $[\omega] \in J^1\pi^*$ be a point, and let
$u = \pi_E^r([\omega])$ be its projection to $E$. Then, since $\Phi$ is a complete solution, there
exists $\lambda_o \in U$ such that $\Phi(\lambda_o,u) \equiv s_{\lambda_o}(u) = [\omega]$, and the
integral sections of $\D_{\lambda_o}$ through $u$, lifted to $J^1\pi^*$ by $s_{\lambda_o}$, give
the integral sections of $\D_h$ through $[\omega]$.

\subsection{Equivalence with the Lagrangian formalism}

In previous Sections we have stated the (generalized) Hamilton-Jacobi problem for multisymplectic
classical field theories in both the Lagrangian and Hamiltonian formalisms. In the following we
establish a bijective relation between the solution of the Hamilton-Jacobi problem in both
formulations in terms of the restricted Legendre map.

First of all, recall that since the Cartan $m$-form $\Theta_\Lag \in \df^{m}(J^1\pi)$ is
$\pi^1$-semibasic, we can give the following definition.

\begin{definition}
The \textnormal{extended Legendre map} associated with the Lagrangian $\Lag \in \df^{m}(J^1\pi)$ is
the bundle map $\widetilde{\Leg} \colon J^1\pi \to \Lambda^m_2(\Tan^*E)$ over $E$ defined as follows:
\begin{equation*}
(\Theta_\Lag(j^1_x\phi))(Y_1(j^1_x\phi),\ldots,Y_m(j^1_x\phi)) =
(\widetilde{\Leg}(j^1_x\phi))((\Tan_{j^1_x\phi}\pi^1Y_1)(\phi(x)),\ldots,(\Tan_{j^1_x\phi}\pi^1Y_m)(\phi(x))) \, ,
\end{equation*}
where $Y_i \in \vf(J^1\pi)$, and therefore $\Tan\pi^1Y_i \in \vf(E)$.
\end{definition}

This map verifies $\pi_E \circ \widetilde{\Leg} = \pi^1$, that is, $\widetilde{\Leg}$ is a bundle
morphism over $E$. Furthermore, we have that $\widetilde{\Leg}^*\Theta = \Theta_\Lag$ and
$\widetilde{\Leg}^*\Omega = \Omega_\Lag$. Then, bearing in mind the coordinate expressions
\eqref{eqn:LiouvilleForms} of the tautological $m$-form $\Theta \in \df^{m}(\Lambda^m_2(\Tan^*E))$,
and \eqref{eqn:CartanMFormLocal} of the Cartan $m$-form $\Theta_\Lag \in \df^{m}(J^1\pi)$,
the coordinate expression of the extended Legendre map is
\begin{equation}\label{eqn:ExtendedLegendreMapLocal}
\widetilde{\Leg}^*x^i = x^i \quad ; \quad \widetilde{\Leg}^*u^\alpha = u^\alpha \quad ; \quad
\widetilde{\Leg}^*p_\alpha^i = \derpar{L}{u_i^\alpha} \quad ; \quad
\widetilde{\Leg}^*p = L - u_i^\alpha \derpar{L}{u_i^\alpha} \, .
\end{equation}

The composition of the extended Legendre map $\widetilde{\Leg}\colon J^1\pi\to\Lambda^{m}_{2}(\Tan^*E)$
with the natural quotient map $\mu \colon \Lambda^m_2(\Tan^*E) \to J^1\pi^*$ gives rise to a bundle
morphism $\mu \circ \widetilde{\Leg} \colon J^1\pi \to J^1\pi^*$, which leads to the following definition.

\begin{definition}
The \textnormal{restricted Legendre map} associated to the Lagrangian density
$\Lag \in \df^{m}(J^1\pi)$ is the bundle morphism $\Leg \colon J^1\pi \to J^1\pi^*$
over $E$ defined as $\Leg = \mu \circ \widetilde{\Leg}$.
\end{definition}

In the natural coordinates of $J^1\pi^*$, the local expression of the restricted Legendre map is
\begin{equation*}
\Leg^*x^i = x^i \quad ; \quad \Leg^*u^\alpha = u^\alpha \quad ; \quad
\Leg^*p_\alpha^i = \derpar{L}{u_i^\alpha} \, .
\end{equation*}

As for the extended Legendre map, the map $\Leg \colon J^1\pi \to J^1\pi^*$ satisfies
$\pi_E^r \circ \Leg = \pi^1$, $\Leg^*\Theta_h = \Theta_\Lag$ and $\Leg^*\Omega_h = \Omega_\Lag$.
Moreover, the Lagrangian density $\Lag \in \df^{m}(J^1\pi)$ is regular if, and only if, the
restricted Legendre map is a local diffeomorphism, and $\Lag$ is said to be \textsl{hyperregular}
if $\Leg$ is a global diffeomorphism. Furthermore, in the hyperregular case the Hamiltonian section
$h \in \Gamma(\mu)$ is defined by $h = \widetilde{\Leg} \circ \Leg^{-1}$ (or by restriction on the
corresponding open sets where $\Leg$ is a diffeomorphism in the regular case).

Finally, to close the review on the properties of the Legendre maps, we have the following
fundamental result from \cite{art:Roman09}.

\begin{theorem}\label{CFTEquivalenceLagrangianHamiltonian}
Let $(J^1\pi,\Lag)$ be a Lagrangian field theory with $\Lag \in \df^{m}(J^1\pi)$ a hyperregular
Lagrangian density, and $(J^1\pi^*,\Omega_h)$ the associated Hamiltonian field theory.
\begin{enumerate}
\item If $\phi \in \Gamma(\pi)$ is a solution to equation \eqref{eqn:LagFieldEqSect}, then the
section $\psi_h = \Leg \circ j^1\phi \in \Gamma(\bar{\pi}^r_E)$ is a solution to equation
\eqref{eqn:HamFieldEqSect}.
\item Conversely, if $\psi_h \in \Gamma(\bar{\pi}^r_E)$ is a solution to equation
\eqref{eqn:HamFieldEqSect}, then the section $\phi = \pi^r_E \circ \psi_h \in \Gamma(\pi)$ is a
solution to equation \eqref{eqn:LagFieldEqSect}.
\end{enumerate}
\end{theorem}

In particular, Theorem \ref{CFTEquivalenceLagrangianHamiltonian} states that the distributions
$\D_\Lag$ and $\D_h$ solutions to the Lagrangian and Hamiltonian problems, respectively, are
$\Leg$-related and, moreover, it establishes a bijective correspondence between the integral sections
of both distributions.

Now we can state the equivalence between the solutions of the (generalized) Lagrangian and Hamiltonian
Hamilton-Jacobi problems in terms of the restricted Legendre map. First we need the following
technical results.

\begin{lemma}\label{lemma:TechLemma1}
Let $E_1 \stackrel{\pi_1}{\longrightarrow} M$ and $E_2 \stackrel{\pi_2}{\longrightarrow} M$
be two fiber bundles, $F \colon E_1 \to E_2$ a fiber bundle morphism, and two $F$-related
$k$-multivector fields $\X_1 \in \vf^{k}(E_1)$ and $\X_2 \in \vf^{k}(E_2)$.
If $s_1 \in \Gamma(\pi_1)$ is a section of $\pi_1$ and we define
a section of $\pi_2$ as $s_2 = F \circ s_1 \in \Gamma(\pi_2)$, then
$$
\Lambda^k\Tan\pi_1 \circ \X_1 \circ s_1 = \Lambda^k\Tan\pi_2 \circ \X_2 \circ s_2 \in \vf^{k}(M) \, .
$$
\end{lemma}
\begin{proof}
As $F \colon E_1 \to E_2$ is a fiber bundle morphism (that is, $\pi_1 = \pi_2 \circ F$), and
$\X_1$ and $\X_2$ are $F$-related (that is, $\Lambda^{k}\Tan F \circ \X_1 = \X_2 \circ F$), we
have the following commutative diagram
$$
\xymatrix{
\Lambda^k\Tan E_1 \ar[rr]^{\Lambda^k\Tan F} &  \ & \Lambda^k\Tan E_2 \\
E_1 \ar[dr]_{\pi_1} \ar[u]^{\X_1} \ar[rr]^{F} & \ & E_2 \ar[dl]^{\pi_2} \ar[u]_{\X_2} \\
\ & M & \
}
$$
Then we have
\begin{align*}
\Lambda^k\Tan\pi_1 \circ \X_1 \circ s_1 &= \Lambda^k\Tan(\pi_2 \circ F) \circ \X_1 \circ s_1
= \Lambda^k\Tan\pi_2 \circ \Lambda^k\Tan F \circ \X_1 \circ s_1 \\
&= \Lambda^k\Tan\pi_2 \circ \X_2 \circ F \circ s_1
= \Lambda^k\Tan\pi_2 \circ \X_2 \circ s_2 \, .
\end{align*}
\end{proof}

\begin{lemma}\label{lemma:TechLemma2}
Let $E_1 \stackrel{\pi_1}{\longrightarrow} M$ and $E_2 \stackrel{\pi_2}{\longrightarrow} M$
be two fiber bundles, $F \colon E_1 \to E_2$ a fiber bundle morphism, and two $F$-related
$k$-dimensional distributions $\D_1$ in $E_1$ and $\D_2$ in $E_2$. If $s_1 \in \Gamma(\pi_1)$
is a section of $\pi_1$ and we define a section of $\pi_2$ as $s_2 = F \circ s_1 \in \Gamma(\pi_2)$,
then
$$
\Tan\pi_1(\restric{\D_1}{\Im(s_1)}) = \Tan\pi_2(\restric{\D_2}{\Im(s_2)}) \, .
$$
\end{lemma}
\begin{proof}
Since $\D_1$ and $\D_2$ are $F$-related, it suffices to apply Lemma \ref{lemma:TechLemma1} to
every pair of $F$-related multivector fields in the classes of multivector fields associated to
$\D_1$ and $\D_2$.
\end{proof}

Finally, the equivalence Theorem is the following.

\begin{theorem}\label{thm:EquivalenceLagHam}
Let $(J^1\pi,\Lag)$ be a Lagrangian field theory with $\Lag \in \df^{m}(J^1\pi)$ a hyperregular
Lagrangian density, and $(J^1\pi^*,\Omega_h)$ its associated Hamiltonian field theory.

\begin{enumerate}
\item If $\Psi \in \Gamma(\pi^1)$ is an integrable jet field solution to the (generalized) Lagrangian
Hamilton-Jacobi problem, then the section $s = \Leg \circ \Psi \in \Gamma(\pi_E^r)$ is a solution to
the (generalized) Hamiltonian Hamilton-Jacobi problem.

\item Conversely, if $s \in \Gamma(\pi_E^r)$ is a solution to the (generalized) Hamiltonian
Hamilton-Jacobi problem, then the jet field $\Psi = \Leg^{-1} \circ s \in \Gamma(\pi^1)$ is
a solution to the (generalized) Lagrangian Hamilton-Jacobi problem.
\end{enumerate}
\end{theorem}
\begin{proof}
This proof follows the patterns of Theorem 3 in \cite{art:Carinena_Gracia_Marmo_Martinez_Munoz_Roman06}
and Theorem 1 in \cite{art:Colombo_DeLeon_Prieto_Roman14_JPA}, now in terms of distributions.

In particular, assume that $\Psi \in \Gamma(\pi^1)$ is a solution to the generalized Lagrangian
Hamilton-Jacobi problem. First, if $\D = \Tan\pi^1(\restric{\D_\Lag}{\Im(\Psi)})$,
$\bar{\D} = \Tan\pi_E^r(\restric{\D_h}{\Im(s)})$ are the integrable distributions associated to
$\Psi$ and $s = \Leg \circ \Psi \in \Gamma(\pi_E^r)$, respectively, then by Lemma
\ref{lemma:TechLemma2} we have $\D = \bar{\D}$. Hence, both distributions are denoted by $\D$.

Now we prove that $s = \Leg \circ \Psi \in \Gamma(\pi_E^r)$ is a solution to the generalized
Hamiltonian Hamilton-Jacobi problem. Let $\gamma \in \Gamma(\pi)$ be an integral section of
$\D$. Computing, we have
\begin{align*}
\gamma^*\inn(Y)\d(h \circ s) &= \gamma^*\inn(Y)\d(h \circ s) = \gamma^*\inn(Y)(s^*\Omega_h) \\
&= \gamma^*\inn(Y)((\Leg \circ \Psi)^*\Omega_h) = \gamma^*\inn(Y)\Psi^*(\Leg^*\Omega_h) \\
&= \gamma^*\inn(Y)(\Psi^*\Omega_\Lag) \, .
\end{align*}
Hence, since this last expression vanishes for every $Y \in \vf(E)$ by the hypothesis and Proposition
\ref{prop:LagGenHJEquivalences}, we have proved
$$
\gamma^*\inn(Y)\d(h \circ s) = 0 \, , \ \mbox{for every } Y \in \vf(E) \, ,
$$
which, using Proposition \ref{prop:HamGenHJEquivalences}, is equivalent to $s = \Leg \circ \Psi$
being a solution to the generalized Hamiltonian Hamilton-Jacobi problem.

Now we require, in addition, $\Psi^*\Omega_\Lag = 0$, that is, $\Psi$ is a solution to the Lagrangian
Hamilton-Jacobi problem. Then we have
$$
s^*\Omega_h = (\Leg \circ \Psi)^*\Omega_h = \Psi^*(\Leg^*\Omega_h) = \Psi^*\Omega_\Lag = 0 \, ,
$$
and therefore $s = \Leg \circ \Psi$ is a solution to the Hamiltonian Hamilton-Jacobi problem.

The converse is proved analogously, but using $\Leg^{-1}$ instead of $\Leg$.
\end{proof}

\section{Examples}
\label{sec:Examples}

\subsection{Non-autonomous dynamical systems}
\label{sec:Example1}

Let us consider the case of non-autonomous dynamical systems, that is, the base manifold $M$ is
$1$-dimensional, and in particular we assume that $M = \R$ with the canonical volume form
$\eta \in \df^{1}(\R)$. Let us consider a first-order non-autonomous dynamical system with $n$
degrees of freedom, and let $\pi \colon E \to \R$ be the configuration bundle for this system,
with $\dim E = n + 1$. The dynamical information is enclosed on a Lagrangian density
$\Lag \in \df^{1}(J^1\pi)$, which is a $\bar{\pi}^1$-semibasic $1$-form. Because of this, let us
denote by $L \in \Cinfty(J^1\pi)$ the function satisfying $\Lag = L \cdot (\bar{\pi}^1)^*\eta$.

\begin{remark}
Observe that since $M$ is $1$-dimensional, it is either diffeomorphic to the real line $\R$ or the unit
circle $\mathbb{S}^1$. The only difference for our calculations is that, contrary to the case of the
unit circle, $\R$ has a canonical global atlas with a distinguished coordinate. Nevertheless, both
manifolds are orientable and parallelizable, and therefore all the calculations remain the same in
the unit circle despite the absence of a global atlas.
\end{remark}

Local coordinates along this Section are denoted in the usual way for non-autonomous dynamical
systems. In particular, let $(t)$ denote the global coordinate in $\R$, and let $(t,q^A)$,
$1 \leqslant A \leqslant n$, be local coordinates in $E$ adapted to the bundle structure. Then,
the induced local coordinates in $J^1\pi$, $\Tan^*E$ and $J^1\pi^*$ are denoted $(t,q^A,v^A)$,
$(t,q^A,p,p_A)$ and $(t,q^A,p_A)$, respectively.

\subsubsection*{Lagrangian formalism}

The Lagrangian problem for first-order non-autonomous dynamical systems consists in finding a class
of $\bar{\pi}^1$-transverse and semi-holonomic vector fields $\{ X_\Lag \} \subseteq \vf(J^1\pi)$
satisfying the field equation \eqref{eqn:LagFieldEqMVF}, that is,
$$
\inn(X_\Lag)\Omega_\Lag = 0 \, , \ \mbox{for every } X_\Lag \in \left\{ X_\Lag \right\} \, .
$$
If the Lagrangian density is regular, the Cartan $2$-form $\Omega \in \df^{2}(J^1\pi)$ has maximal
rank $2n$ on $J^1\pi$, and therefore there exists a class of vector fields $\{ X_\Lag \}$ solution
to the above equation which is holonomic. Hence, in the following we assume that the Lagrangian
density is regular.

\begin{remark}
Note that the integrability is assured without further assumptions in this setting, since every
vector field defined on a manifold is always integrable.
\end{remark}

The local expressions of the Cartan $1$-form $\Theta_\Lag \in \df^{1}(J^1\pi)$ is
\begin{equation}\label{eqn:Example1_Cartan1Form}
\Theta_\Lag = \derpar{L}{v^A}\d q^A - \left( v^A \derpar{L}{v^A} - L \right) \d t \, ,
\end{equation}
from where we deduce the coordinate expression of the Cartan $2$-form, which is
\begin{equation}\label{eqn:Example1_Cartan2Form}
\begin{array}{l}
\displaystyle \Omega_\Lag = \derpars{L}{v^A}{q^B}\d q^A \wedge \d q^B + \derpars{L}{v^A}{v^B} \d q^A \wedge \d v^B \\[12pt]
\displaystyle \qquad\qquad + \left( v^A\derpars{L}{v^A}{q^B} - \derpar{L}{q^B} \right) \d q^B \wedge \d t + v^A\derpars{L}{v^A}{v^B} \d v^B \wedge \d t \, .
\end{array}
\end{equation}
Thus, a representative of the equivalence class $\{ X_\Lag \} \subseteq \vf(J^1\pi)$ of
$\bar{\pi}^1$-transverse vector fields solution to the above dynamical equation is given
in coordinates by
\begin{equation}\label{eqn:Example3_LagrangianVF}
X_\Lag = f\left( \derpar{}{t} + v^A\derpar{}{q^A} + F^A\derpar{}{v^A} \right) \, ,
\end{equation}
where $f$ is a non-vanishing local function, and the $n$ functions $F^A$ are the unique solutions
to the Euler-Lagrange equations
$$
\derpar{L}{q^B}  - \derpars{L}{t}{v^B} - v^A\derpars{L}{q^A}{v^B} - F^A\derpars{L}{v^A}{v^B} = 0 \, .
$$
Observe that the vector fields in the class $\{ X_\Lag \} \subseteq \vf(J^1\pi)$ are completely
determinated (except for the non-vanishing function $f$). In the following we take $f = 1$ as a
representative of the equivalence class to state the Lagrangian Hamilton-Jacobi problem.

For the generalized Lagrangian Hamilton-Jacobi problem, all the Definitions and results stated in
Section \ref{sec:LagGenHJProblem} remain without changes, except for Corollary
\ref{corol:LagGenHJProjectedSections}, where ``boundary conditions'' should be replaced by ``initial
conditions''. Note that every mention to the integrability of the jet field is redundant, since
the associated distribution in this case is $1$-dimensional, and therefore integrable without
further assumptions.

In coordinates, let $\Psi \in \Gamma(\pi^1)$ be a jet field locally given by∞
$\Psi(t,q^A) = (t,q^A,\psi^A(t,q^A))$, with $\psi^A \in \Cinfty(E)$ being local functions. Then,
from Proposition \ref{prop:LagGenHJEquivalences} we know that $\Psi$ is a solution to the
generalized Lagrangian Hamilton-Jacobi problem if, and only if, every vector field in the class
$\{ X_\Lag \}$ is tangent to the submanifold $\Im(\Psi) \hookrightarrow J^1\pi$, which is locally
defined by the constraints $v^A - \psi^A = 0$, $1 \leqslant A \leqslant n$. Then, the conditions
$\Lie(X_\Lag)(v^A - \psi^A) = 0$ give rise to the following system of $n$ partial differential
equations
$$
\restric{F^A}{\Im(\Psi)} - \derpar{\psi^A}{t} - \psi^B\derpar{\psi^A}{q^B} = 0 \, .
$$

Now, in order to state the Lagrangian Hamilton-Jacobi problem, we must require in addition the jet
field $\Psi \in \Gamma(\pi^1)$ to satisfy the condition $\Psi^*\Omega_\Lag = 0$. In this situation,
Definition \ref{def:LagHJDef} remains unchanged, but the statement of Proposition
\ref{prop:LagHJEquivalences} must be changed as follows.

\begin{proposition}\label{prop:Example1_LagHJEquivalences}
Let $\Psi \in \Gamma(\pi^1)$ be a jet field satisfying $\Psi^*\Omega_\Lag = 0$.
Then the following statements are equivalent:
\begin{enumerate}
\item $\Psi$ is a solution to the Lagrangian Hamilton-Jacobi problem.
\item The submanifold $\Im(\Psi) \hookrightarrow J^1\pi$ is Lagrangian with respect to the
presymplectic form $\Omega_\Lag$ and every vector field $X_\Lag \in \{ X_\Lag \}$ is tangent to it.
\item The integral curves of $X_\Lag \in \{ X_\Lag \}$ with initial conditions in $\Im(\Psi)$
project onto the integral curves of $X = \Tan\pi^1 \circ X_\Lag \circ \Psi$.
\end{enumerate}
\end{proposition}

Bearing in mind the coordinate expression \eqref{eqn:Example1_Cartan2Form} of the Cartan $2$-form,
the pull-back $\Psi^*\Omega_\Lag \in \df^{2}(E)$ has the following coordinate expression
\begin{align*}
\Psi^*\Omega_\Lag &= \left( \derpars{L}{v^A}{q^B} + \derpars{L}{v^A}{v^C}\derpar{\psi^C}{q^B} \right) \d q^A \wedge \d q^B \\
&\quad{} + \left( \derpars{L}{v^B}{v^A}\derpar{\psi^A}{t} + \psi^A \derpars{L}{v^A}{q^B} - \derpar{L}{q^B} +
\psi^A \derpars{L}{v^A}{v^C} \derpar{\psi^C}{q^B} \right) \d q^B \wedge \d t \, .
\end{align*}
Hence, the condition $\Psi^*\Omega_\Lag = 0$ gives the following system of partial differential
equations
$$
\derpars{L}{v^A}{q^B} + \derpars{L}{v^A}{v^C}\derpar{\psi^C}{q^B} = 0 \quad ; \quad
\derpars{L}{v^B}{v^A}\derpar{\psi^A}{t} + \psi^A \derpars{L}{v^A}{q^B} - \derpar{L}{q^B} +
\psi^A \derpars{L}{v^A}{v^C} \derpar{\psi^C}{q^B} = 0 \, ,
$$
which may be combined to obtain equations \eqref{eqn:LagHJLocal} in this setting, that is,
the following system of $n^2$ partial differential equations
$$
\derpars{L}{v^A}{q^B} + \derpar{\psi^C}{q^B}\derpars{L}{v^A}{v^C} = 0 \quad ; \quad
\derpars{L}{v^B}{v^A}\derpar{\psi^A}{t} - \derpar{L}{q^B} = 0 \, .
$$

Finally, we state the Hamilton-Jacobi equation in the Lagrangian formalism. Since the jet field
$\Psi \in \Gamma(\pi^1)$ is a solution to the Lagrangian Hamilton-Jacobi problem, we have that
$\d(\Psi^*\Theta_\Lag) = 0$. Thus, there exists a local function $W \in \Cinfty(E)$ such that
$\d W = \Psi^*\Theta_\Lag$. In coordinates, the $1$-form $\Psi^*\Theta_\Lag$ is given by
$$
\Psi^*\Theta_\Lag = \derpar{L}{v^A}\d q^A - \left( \psi^A \derpar{L}{v^A} - L(t,q^A,\psi^A) \right) \d t \, .
$$
Hence, requiring $\Psi^*\Theta_\Lag = \d W$, we obtain
$$
\derpar{W}{t} = - \psi^A\derpar{L}{v^A} + L(t,q^A,\psi^A) \quad ; \quad
\derpar{W}{q^A} = \derpar{L}{v^A} \, ,
$$
which may be combined in the following single equation
$$
\derpar{W}{t} + \psi^A\derpar{W}{q^A} - L(t,q^A,\psi^A) = 0 \, .
$$
This is the usual Hamilton-Jacobi equation for a first-order non-autonomous Lagrangian dynamical
system.

\subsubsection*{Hamiltonian formalism}

Now we state the Hamiltonian formulation of the Hamilton-Jacobi problem in this setting. In the
natural coordinates of $J^1\pi^*$, the coordinate expression of the restricted Legendre map
$\Leg \colon J^1\pi \to J^1\pi^*$ in this setting is the following
$$
\Leg^*p_A = \derpar{L}{v^A} \, ,
$$
from where we deduce the coordinate expression of the extended Legendre map, which is
$$
\widetilde{\Leg}^*p_A = \derpar{L}{v^A} \quad ; \quad
\widetilde{\Leg}^*p = L - v^A \derpar{L}{v^A} \, .
$$
Recall that since $\Lag \in \df^{1}(J^1\pi)$ is assumed to be regular, the restricted Legendre map
is a local diffeomorphism. For simplicity, in the following we assume that $\Lag$ is hyperregular,
and thus $\Leg$ is a global diffeomorphism (the regular but not hyperregular case is recovered by
restriction on the open sets where the restricted Legendre map is a diffeomorphism).
Let $h = \widetilde{\Leg} \circ \Leg^{-1} \in \Gamma(\mu)$ be the Hamiltonian section associated
to the Lagrangian density $\Lag$, and $H \in \Cinfty(J^1\pi)$ a local Hamiltonian function which
specifies this Hamiltonian section $h$.

Using the Hamiltonian section $h \in \Gamma(\mu)$ we define the Hamilton-Cartan forms
$\Theta_h = h^*\Theta \in \df^{1}(J^1\pi^*)$ and $\Omega_h = h^*\Omega \in \df^{2}(J^1\pi^*)$,
with coordinate expressions
\begin{equation}\label{eqn:Example1_HamiltonCartan2Form}
\Theta_h = p_A\d q^A - H\d t \quad ; \quad
\Omega_h = \d q^A \wedge \d p_A + \d H \wedge \d t \, .
\end{equation}

The Hamiltonian problem for first-order non-autonomous dynamical systems consists in finding a class
of $\bar{\pi}^r_E$-transverse vector fields $\{ X_h \} \subseteq \vf(J^1\pi^*)$ satisfying equation
\eqref{eqn:HamFieldEqMVF}, that is,
$$
\inn(X_h)\Omega_h = 0 \, , \ \mbox{for every } X_h \in \left\{ X_h \right\} \, .
$$
As in the general setting described at the beginning of Section \ref{sec:HamiltonianFormalism},
in this formulation the $2$-form $\Omega_h$ in $J^1\pi^*$ has maximal rank $2n$ regardless of the
Hamiltonian section $h \in \Gamma(\mu)$. Therefore, there always exists such a class of vector fields.

\begin{remark}
As in the Lagrangian formulation described previously, the integrability is assured without further
assumptions in this setting.
\end{remark}

Then, a representative of the equivalence class $\{ X_h \} \subseteq \vf(J^1\pi^*)$ of
$\bar{\pi}_E^r$-transverse vector fields solution to the above dynamical equation is given
in coordinates by
$$
X_h = f \left( \derpar{}{t} + \derpar{H}{p_A} \derpar{}{q^A} - \derpar{H}{q^A} \derpar{}{p_A} \right) \, .
$$
As in the Lagrangian formulation, note that the vector fields in the class
$\{ X_h \} \subseteq \vf(J^1\pi^*)$ are completely determinated (again, except for the non-vanishing
function $f$). In the following we take $f = 1$ as a representative of the equivalence class to state
the Hamiltonian Hamilton-Jacobi problem.

For the generalized Hamiltonian Hamilton-Jacobi problem, all the Definitions and results stated in
Section \ref{sec:HamGenHJProblem} remain unchanged, except for Corollary
\ref{corol:HamGenHJProjectedSections}, where ``boundary conditions'' should be replaced by ``initial
conditions''.

In coordinates, let $s \in \Gamma(\pi_E^r)$ be a section locally given by $s(t,q^A) = (t,q^A,s_A(t,q^A))$,
where $s_A \in \Cinfty(E)$ are local functions. Then, from Proposition \ref{prop:HamGenHJEquivalences}
we know that $s$ is a solution to the generalized Hamiltonian Hamilton-Jacobi problem if, and only
if, every vector field in the class $\{ X_h \}$ is tangent to the submanifold
$\Im(s) \hookrightarrow J^1\pi^*$. Since this submanifold is defined locally by the constraints
$p_A - s_A = 0$, $1 \leqslant A \leqslant n$, we must check if the conditions
$\Lie(X_h)(p_A - s_A) = 0$ hold, which give rise the following system of $n$ partial differential
equations
$$
\derpar{H}{q^A} + \derpar{s_A}{t} + \derpar{H}{p_B}\derpar{s_A}{q^B} = 0 \, , \quad \mbox{(on $\Im(s)$)} \, .
$$

Now, in order to state the Hamiltonian Hamilton-Jacobi problem, we must require in addition the
section $s \in \Gamma(\pi_E^r)$ to satisfy the condition $s^*\Omega_h = 0$. In this case, the
statement in Definition \ref{def:HamHJDef} and the remark that follows remain unchanged, but the
statement of Proposition \ref{prop:HamHJEquivalences} must be changed as follows

\begin{proposition}\label{prop:Example1_HamHJEquivalences}
Let $s \in \Gamma(\pi_E^r)$ be a section satisfying $s^*\Omega_h = 0$. Then the following
statements are equivalent:
\begin{enumerate}
\item $s$ is a solution to the Hamiltonian Hamilton-Jacobi problem.
\item The submanifold $\Im(s) \hookrightarrow J^1\pi^*$ is Lagrangian with respect to the
presymplectic form $\Omega_h$ and every vector field $X_h \in \{ X_h \}$ is tangent to it.
\item The integral curves of $X_h \in \{ X_h \}$ with initial conditions in $\Im(s)$
project onto the integral curves of $X = \Tan\pi_E^r \circ X_h \circ s$.
\end{enumerate}
\end{proposition}

In coordinates, bearing in mind the coordinate expression \eqref{eqn:Example1_HamiltonCartan2Form}
of the Hamilton-Cartan $2$-form $\Omega_h$, the pull-back $s^*\Omega_h \in \df^{2}(E)$ has the
following local expression
$$
s^*\Omega_h = \left( \derpar{s_A}{t} + \derpar{H}{q^A} + \derpar{H}{p_B}\derpar{s_B}{q^A} \right) \d q^A \wedge \d t
+ \derpar{s_A}{q^B} \d q^A \wedge \d q^B \, .
$$
Hence, the condition $s^*\Omega_h = 0$ gives equations \eqref{eqn:HamHJLocal}, which in this case
correspond to the following system of $n^2$ partial differential equations
$$
\derpar{s_A}{t} + \derpar{H}{q^A} + \derpar{H}{p_B}\derpar{s_B}{q^A} \quad ; \quad
\derpar{s_A}{q^B} - \derpar{s_B}{q^A} = 0 \, .
$$

Finally, we state the Hamilton-Jacobi equation in the Hamiltonian formalism. Since the section
$s \in \Gamma(\pi_E^r)$ is a solution to the Lagrangian Hamilton-Jacobi problem, we have that
$h \circ s \in \df^{1}(E)$ is a closed form. Therefore, using Poincar\'e's Lemma, there exists a
local function $W \in \Cinfty(E)$ such that $\d W = h \circ s$. Bearing in mind that the $1$-form
$h \circ s \in \df^{1}(E)$ is given in coordinates by
$$
h \circ s = s_A \d q^A - (H \circ s) \d t \, ,
$$
the condition $h \circ s = \d W$ gives the following partial differential equations
$$
\derpar{W}{t} = -H(t,q^A,s_A) \quad ; \quad
\derpar{W}{q^A} = s_A \, ,
$$
which may combined to obtain the following single equation
$$
\derpar{W}{t} + H\left( t,q^A,\derpar{W}{q^A} \right) = 0 \, .
$$
This is the classic Hamilton-Jacobi equation for a first-order non-autonomous dynamical system
with $n$ degrees of freedom and Hamiltonian function $H$.

\subsection{Quadratic Lagrangian densities}
\label{sec:Example2}

Let us consider a classical field theory with $n$ fields depending on $m$ independent variables
given in terms of a quadratic Lagrangian density, and let $\pi \colon E \to M$ be the configuration
bundle, with $M$ being a $m$-dimensional orientable smooth manifold with fixed volume form
$\eta \in \df^{m}(M)$, and $\dim E = m + n$. Most of the quadratic Lagrangian field theories
can be modeled as follows \cite{book:Sardanashvily95,book:Giachetta_Mangiarotti_Sardanashvily97}:
let us assume that the bundle $\pi \colon E \to M$ is trivial (that is, $E = M \times Q$, with
$\dim Q = n$), so that $\pi^1 \colon J^1\pi \to E$ is a vector bundle. Let $g$ be a pseudo-Riemannian
metric in this vector bundle, $\Gamma$ a connection of the projection $\pi^1$ and $V \in \Cinfty(E)$
a potential function. Then a quadratic Lagrangian density $\Lag \in \df^{1}(J^1\pi)$ is given by
$$
\Lag(j^1_x\phi) = \left(
\frac{1}{2} \, g(j^1_x\phi - (\Gamma \circ \pi^1)(j^1_x\phi)), j^1_x\phi - (\Gamma \circ \pi^1)(j^1_x\phi)))
+ ((\pi^1)^*V)(j^1_x\phi) \right) (\bar{\pi}^1)^*\eta \, .
$$
In the natural coordinates $(x^i,u^\alpha,u_i^\alpha)$ of $J^1\pi$, the $m$-form $\Lag$ has the
following local expression
$$
\Lag(x^i,u^\alpha,u_i^\alpha) =
\left( \frac{1}{2} \, g^{ij}_{\alpha\beta}(u_i^\alpha - \Gamma_i^\alpha)(u_j^\beta - \Gamma_j^\beta) + V(x^i,u^\alpha) \right) \d^mx \, ,
$$
where $g^{ij}_{\alpha\beta} = g^{ij}_{\alpha\beta}(x^i,u^\alpha)$ are the coefficients of the metric,
which moreover satisfy $g^{ij}_{\alpha\beta} = g^{ji}_{\beta\alpha}$ due to the symmetry assumption,
and $\Gamma_i^\alpha = \Gamma_i^\alpha(x^i)$ are the component functions of the connection.

In order to simplify the problem, we assume that the connection $\Gamma$ is integrable. As a
consequence, there exist natural charts of coordinates in $J^1\pi$ such that $\Gamma^\alpha_i = 0$,
which implies that the coordinate expression of the Lagrangian density $\Lag$ reduces to
\begin{equation}\label{eqn:Example2_LagrangianLocal}
\Lag(x^i,u^\alpha,u_i^\alpha) =
\left( \frac{1}{2} \, g^{ij}_{\alpha\beta}u_i^\alpha u_j^\beta + V(x^i,u^\alpha) \right) \d^mx \equiv L \cdot \d^mx \, .
\end{equation}

Observe that the Hessian matrix of the Lagrangian function $L$ associated with $\Lag$ and $\eta$
coincides with the matrix of the coefficients of $g$, that is,
$$
\left( \derpars{L}{u_i^\alpha}{u_j^\beta} \right) = \left( g^{ij}_{\alpha\beta} \right) \, ,
$$
and therefore the Lagrangian density $\Lag$ is regular since $g$ is a pseudo-Riemannian metric, and
in particular non-degenerate.

\subsubsection*{Lagrangian formalism}

The local expression of the Cartan $m$-form $\Theta_\Lag \in \df^{m}(J^1\pi)$ obtained from the
Lagrangian density given locally by \eqref{eqn:Example2_LagrangianLocal} is
\begin{equation} \label{eqn:Example2_CartanMForm}
\Theta_\Lag = g_{\alpha\beta}^{ij} u_j^\beta \d u^\alpha \wedge \d^{m-1}x_i -
\left( \frac{1}{2} \, g^{ij}_{\alpha\beta}u_i^\alpha u_j^\beta - V(x^i,u^\alpha) \right)\d^{m}x \, ,
\end{equation}
from where we obtain the coordinate expression of the Cartan $(m+1)$-form, which is
\begin{equation}\label{eqn:Example2_CartanM+1Form}
\begin{array}{l}
\displaystyle \Omega_\Lag =
- \derpar{g^{ij}_{\alpha\beta}}{u^\delta} \, u_j^\beta \d u^\delta \wedge \d u^\alpha \wedge \d^{m-1}x_i
- g^{ij}_{\alpha\beta} \d u_j^\beta \wedge \d u^\alpha \wedge \d^{m-1}x_i
+ g^{ij}_{\alpha\beta}u_i^\alpha \d u_j^\beta \wedge \d^mx \\[10pt]
\displaystyle\qquad\quad{} + \left( \derpar{g^{ij}_{\delta\beta}}{x^i} \, u_j^\beta
+ \frac{1}{2}\,\derpar{g^{ij}_{\alpha\beta}}{u^\delta} u_i^\alpha u_j^\beta
- \derpar{V}{u^\delta} \right) \d u^\delta \wedge \d^mx \, .
\end{array}
\end{equation}
Hence, a representative of the class of holonomic multivector fields $\{ \X_\Lag \} \subseteq \vf^{m}(J^1\pi)$
solution to the Lagrangian field equation \eqref{eqn:LagFieldEqMVF} is given in coordinates by
$$
\X_\Lag = \bigwedge_{j=1}^{m} \left( \derpar{}{x^j} + u_j^\alpha\derpar{}{u^\alpha}
+ F_{j,i}^\alpha \derpar{}{u_i^\alpha} \right) \, ,
$$
where the functions $F_{j,i}^\alpha$ are solutions to the Euler-Lagrange equations for multivector
fields \eqref{eqn:LagFieldEqLocal}, which in this case are
\begin{equation}\label{eqn:Example2_EulerLagrangeEqMVF}
\frac{1}{2}\,\derpar{g_{\delta\beta}^{ij}}{u^\alpha} \, u_i^\delta u_j^\beta + \derpar{V}{u^\alpha}
- \sum_{i=1}^{m} \derpar{g_{\alpha\beta}^{ij}}{x^i} u_j^\beta
- \derpar{g_{\alpha\delta}^{ij}}{u^\beta} u_i^\beta u_j^\delta
- F_{j,i}^\beta g_{\alpha\beta}^{ij} = 0 \, ,
\end{equation}
in addition to the integrability conditions \eqref{eqn:Integrability&SemiHolonomicity} (if necessary).

First we state the generalized version of the Hamilton-Jacobi problem. Let $\Psi \in \Gamma(\pi^1)$
be a jet field locally given by $\Psi(x^i,u^\alpha) = (x^i,u^\alpha,\psi_i^\alpha(x^i,u^\alpha))$,
where $\psi_i^\alpha \in \Cinfty(E)$ are local functions. Then, from Proposition
\ref{prop:LagGenHJEquivalences} we know that $\Psi$ is a solution to the generalized Lagrangian
Hamilton-Jacobi problem if, and only if, every multivector field $\X_\Lag \in \{ \X_\Lag \}$ is
tangent to the submanifold $\Im(\Psi) \hookrightarrow J^1\pi$. Bearing in mind that $\Im(\Psi)$ is
defined locally by the constraints $u_k^\alpha - \psi_k^\alpha = 0$, the tangency condition
gives rise to equations \eqref{eqn:LagGenHJLocal}, where now the component functions $F_{j,i}^\alpha$
are solutions to the $n$ equations \eqref{eqn:Example2_EulerLagrangeEqMVF}.

In order to state the equations of the Lagrangian Hamilton-Jacobi problem, we must require in
addition the jet field $\Psi \in \Gamma(\pi^1)$ to satisfy the condition $\Psi^*\Omega_\Lag = 0$.
Bearing in mind the coordinate expression \eqref{eqn:Example2_CartanM+1Form} of the Cartan
$(m+1)$-form $\Omega_\Lag$, the pull-back $\Psi^*\Omega_\Lag \in \df^{m+1}(E)$ is given in
coordinates by
\begin{align*}
\Psi^*\Omega_\Lag =
&- \left( \derpar{g_{\alpha\beta}^{ij}}{u^\delta}\psi^\beta_j
+ g_{\alpha\beta}^{ij} \derpar{\psi_j^\beta}{u^\delta} \right) \d u^\delta \wedge \d ^\alpha \wedge \d^{m-1}x_i \\
&{}+ \left( g_{\delta\beta}^{ij} \derpar{\psi_j^\beta}{x^i}
+ g_{\alpha\beta}^{ij}\psi_i^\alpha\derpar{\psi_j^\beta}{u^\delta}
+ \derpar{g_{\delta\beta}^{ij}}{x^i} \psi_j^\beta
+ \frac{1}{2}\,\derpar{g_{\alpha\beta}^{ij}}{u^\delta}\psi_i^\alpha\psi_j^\beta
- \derpar{V}{u^\delta} \right) \d u^\delta \wedge \d^mx \, .
\end{align*}
Hence, the condition $\Psi^*\Omega_\Lag = 0$ gives the following system of partial differential
equations
$$
\derpar{g_{\alpha\beta}^{ij}}{u^\delta}\psi^\beta_j + g_{\alpha\beta}^{ij} \derpar{\psi_j^\beta}{u^\delta} = 0 \quad ; \quad
g_{\delta\beta}^{ij} \derpar{\psi_j^\beta}{x^i} + g_{\alpha\beta}^{ij}\psi_i^\alpha\derpar{\psi_j^\beta}{u^\delta}
+ \derpar{g_{\delta\beta}^{ij}}{x^i} \psi_j^\beta + \frac{1}{2}\,\derpar{g_{\alpha\beta}^{ij}}{u^\delta}\psi_i^\alpha\psi_j^\beta
- \derpar{V}{u^\delta} = 0 \, ,
$$
which may be combined to obtain equations \eqref{eqn:LagHJLocal}, that is, the following system of
$n(1+m(n-1))$ partial differential equations
$$
\derpar{g_{\alpha\beta}^{ij}}{u^\delta}\psi^\beta_j + g_{\alpha\beta}^{ij} \derpar{\psi_j^\beta}{u^\delta} = 0 \quad ; \quad
g_{\delta\beta}^{ij} \derpar{\psi_j^\beta}{x^i} + \derpar{g_{\delta\beta}^{ij}}{x^i} \psi_j^\beta
- \frac{1}{2}\,\derpar{g_{\alpha\beta}^{ij}}{u^\delta}\psi_i^\alpha\psi_j^\beta
- \derpar{V}{u^\delta} = 0 \, .
$$

Finally, we deduce the Lagrangian Hamilton-Jacobi equation. Since the jet field $\Psi \in \Gamma(\pi^1)$
is a solution to the Lagrangian Hamilton-Jacobi problem, we have that $\Psi^*\Theta_\Lag \in \df^{m}(E)$
is a closed form, and thus there exists a local $\pi$-semibasic $(m-1)$-form $\omega \in \df^{m-1}(E)$
such that $\d\omega = \Psi^*\Theta_\Lag$. In coordinates, bearing in mind the coordinate expression
\eqref{eqn:Example2_CartanMForm} of the Cartan $m$-form, we obtain
$$
\Psi^*\Theta_\Lag = g_{\alpha\beta}^{ij} \psi_j^\beta \d u^\alpha \wedge \d^{m-1}x_i -
\left( \frac{1}{2} \, g^{ij}_{\alpha\beta}\psi_i^\alpha \psi_j^\beta - V(x^i,u^\alpha) \right)\d^{m}x \, .
$$
Then, if the $(m-1)$-form $\omega$ is given locally by $\omega = W^i\d^{m-1}x_i$, the condition
$\d \omega = \Psi^*\Theta_\Lag$ gives rise to the following equations
$$
\sum_{i=1}^{m}\derpar{W^i}{x^i} + \frac{1}{2} \, g^{ij}_{\alpha\beta}\psi_i^\alpha \psi_j^\beta - V(x^i,u^\alpha) = 0
\quad ; \quad
\derpar{W^i}{u^\alpha} = g_{\alpha\beta}^{ij} \psi_j^\beta \, ,
$$
which may be combined to give the classic Hamilton-Jacobi equation in the Lagrangian formalism
\eqref{eqn:LagHJEq}, which in this example is
\begin{equation}\label{eqn:Example2_HJEquation}
\sum_{i=1}^{m}\derpar{W^i}{x^i}
+ \frac{1}{2} \, \tilde{g}_{ij}^{\alpha\beta} \derpar{W^i}{u^\alpha}\,\derpar{W^j}{u^\beta} - V(x^i,u^\alpha) = 0 \, .
\end{equation}
where $\tilde{g}_{ij}^{\alpha\beta}$ denote the coefficients of the inverse matrix of
$\left(g^{ij}_{\alpha\beta}\right)$, which exists since we assume $g$ to be non-degenerate.

\subsubsection*{Hamiltonian formalism}

Now we state the Hamiltonian formalism for the Hamilton-Jacobi problem. First, let us compute the
coordinate expression of the Legendre maps associated to the Lagrangian density given by
\eqref{eqn:Example2_LagrangianLocal}. In the natural coordinates $(x^i,u^\alpha,p_\alpha^i)$ of
$J^1\pi^*$, the restricted Legendre map $\Leg \colon J^1\pi \to J^1\pi^*$ has the following
coordinate expression
$$
\Leg^*p_\alpha^i = g_{\alpha\beta}^{ij} u_j^\beta \, ,
$$
from where the local expression of the extended Legendre map is straightforwardly deduced as
$$
\widetilde{\Leg}^*p_\alpha^i = g_{\alpha\beta}^{ij} u_j^\beta \quad ; \quad
\widetilde{\Leg}^*p = V(x^i,u^\alpha) - \frac{1}{2} \, g^{ij}_{\alpha\beta}u_i^\alpha u_j^\beta \, ,
$$

Using the coordinate expression of both Legendre maps we obtain a local Hamiltonian function
associated to the Lagrangian density $\Lag$, whose coordinate expression is
$$
H(x^i,u^\alpha,p_i^\alpha) = \frac{1}{2} \, \tilde{g}_{ij}^{\alpha\beta} p_\alpha^i p_\beta^j - V(x^i,u^\alpha) \, ,
$$
where, as before, $\left(\tilde{g}_{ij}^{\alpha\beta}\right)$ is the inverse matrix of
$\left(g^{ij}_{\alpha\beta}\right)$.

Using the Hamiltonian section $h \in \Gamma(\mu)$ specified by this local Hamiltonian function
we define the Hamilton-Cartan forms $\Theta_h = h^*\Theta \in \df^{m}(J^1\pi^*)$,
$\Omega_h = -\d \Theta_h \in \df^{m+1}(J^1\pi^*)$, whose coordinate expressions are
\begin{align}\label{eqn:Example2_HamiltonCartanForms}
&\Theta_h = p_\alpha^i\d u^\alpha \wedge \d^{m-1}x_i - \left( \frac{1}{2} \, \tilde{g}_{ij}^{\alpha\beta} p_\alpha^i p_\beta^j - V(x^i,u^\alpha) \right) \d^mx \, , \\
&\Omega_h = -\d p_\alpha^i \wedge \d u^\alpha \wedge \d^{m-1}x_i
+ \left( \frac{1}{2}\,\derpar{\tilde{g}^{\alpha\beta}_{ij}}{u^\delta} p_\alpha^i p_\beta^j - \derpar{V}{u^\delta} \right) \d u^\delta \wedge \d^{m}x
+ \tilde{g}_{ij}^{\alpha\beta} p_\alpha^i \d p_\beta^j \wedge \d^mx \, .
\end{align}

Then, a representative of the class of $\bar{\pi}_E^r$-transverse multivector fields
$\{ \X_h \} \subseteq \vf^{m}(J^1\pi^*)$ which are solutions to the Hamiltonian field equation
\eqref{eqn:HamFieldEqMVF} is given in coordinates by
\begin{equation}\label{eqn:Example2_HamiltonianMVFParticular}
\X_h = \bigwedge_{j=1}^{m} \left( \derpar{}{x^j} + \tilde{g}_{ij}^{\alpha\beta}p_\alpha^i \derpar{}{u^\beta}
+ G_{\alpha,j}^j \derpar{}{p^i_\alpha} \right) \, ,
\end{equation}
with the functions $G_{\alpha,j}^i$ satisfying the Hamilton-De Donder-Weyl equations \eqref{eqn:HamFieldEqLocal},
which in this case are
\begin{equation}\label{eqn:Example2_HamiltonEqMVF}
\sum_{i=1}^{m}G_{\delta,i}^{i} = -\left( \frac{1}{2}\,\derpar{\tilde{g}^{\alpha\beta}_{ij}}{u^\delta} p_\alpha^i p_\beta^j - \derpar{V}{u^\delta} \right) \, ,
\end{equation}
in addition to the integrability conditions \eqref{eqn:Integrability} (if necessary).

Let us state the generalized Hamiltonian Hamilton-Jacobi problem for this field theory. Let
$s \in \Gamma(\pi_{E}^r)$ be a section given locally by
$s(x^i,u^\alpha) = (x^i,u^\alpha,s^i_\alpha(x^i,u^\alpha))$, where $s^i_\alpha \in \Cinfty(E)$ are
local functions. By Proposition \ref{prop:HamGenHJEquivalences}, the section $s$ is a solution to
the generalized Hamiltonian Hamilton-Jacobi problem if, and only if, the multivector field $\X_h$
given in coordinates by \eqref{eqn:Example2_HamiltonianMVFParticular} is tangent to the submanifold
$\Im(s) \hookrightarrow J^1\pi^*$ defined locally by the constraints $p^i_\alpha - s^i_\alpha = 0$.
Computing, the tangency condition gives rise to equations \eqref{eqn:HamGenHJLocal}, that is,
$$
\derpar{s_\delta^k}{x^j} + \tilde{g}_{ij}^{\alpha\beta}p_\alpha^i \derpar{s_\delta^k}{u^\beta}
- \restric{G_{\delta,j}^k}{\Im(s)} = 0 \, .
$$

Now, in order to state the Hamiltonian Hamilton-Jacobi problem, we must require in addition the
section $s \in \Gamma(\pi_E^r)$ to satisfy the condition $s^*\Omega_h = 0$ or, equivalently,
$\d(h \circ s) = 0$. Bearing in mind the coordinate expression \eqref{eqn:Example2_HamiltonCartanForms}
of the Hamilton-Cartan $(m+1)$-form $\Omega_h$, the pull-back $s^*\Omega_h \in \df^{m+1}(E)$ is given
locally by
$$
s^*\Omega_h = \left( \derpar{s_\delta^i}{x^i}
+ \frac{1}{2}\,\derpar{\tilde{g}_{ij}^{\alpha\beta}}{u^\delta}s_\alpha^is_\beta^j - \derpar{V}{u^\delta}
+ \tilde{g}_{ij}^{\alpha\beta} s_\alpha^i \derpar{s_\beta^j}{u^\delta} \right) \d u^\delta \wedge \d^mx
- \derpar{s_\alpha^i}{u^\beta} \d u^\beta \wedge \d u^\alpha \wedge \d^{m-1}x_i \, .
$$
Hence, the condition $s^*\Omega_h = 0$ gives equations \eqref{eqn:HamHJLocal} for this problem,
that is,
$$
\derpar{s_\delta^i}{x^i} + \frac{1}{2}\,\derpar{\tilde{g}_{ij}^{\alpha\beta}}{u^\delta}s_\alpha^is_\beta^j
- \derpar{V}{u^\delta} + \tilde{g}_{ij}^{\alpha\beta} s_\alpha^i \derpar{s_\beta^j}{u^\delta} = 0 \quad ; \quad
\derpar{s_\alpha^i}{u^\beta} - \derpar{s_\beta^i}{u^\alpha} = 0 \, .
$$

Finally, we compute the Hamilton-Jacobi equation in the Hamiltonian formalism. Since the section
$s \in \Gamma(\pi_E^r)$ is a solution to the Hamiltonian Hamilton-Jacobi problem, we have that
$h \circ s \in \df^{m}(E)$ is a closed form. Hence, there exists a local $\pi$-semibasic $(m-1)$-form
$\omega \in \df^{m-1}(E)$ such that $\d\omega = h \circ s$. In coordinates, the local expression of
the $m$-form $h \circ s$ is
$$
h \circ s = s_\alpha^i \d u^\alpha \wedge \d^{m-1}x_i
- \left( \frac{1}{2} \, \tilde{g}_{ij}^{\alpha\beta} s_\alpha^i s_\beta^j - V(x^i,u^\alpha) \right) \d^mx \, ,
$$
Then, if the $(m-1)$-form $\omega$ is given locally by $\omega = W^i\d^{m-1}x_i$, the condition
$\d \omega = h \circ s$ gives rise to the following equations
$$
\sum_{i=1}^{m}\derpar{W^i}{x^i} + \frac{1}{2} \, \tilde{g}_{ij}^{\alpha\beta} s_\alpha^i s_\beta^j - V(x^i,u^\alpha) = 0
\quad ; \quad
\derpar{W^i}{u^\alpha} = s^i_\alpha \, ,
$$
which may be combined to give the classic Hamilton-Jacobi equation in the Hamiltonian formalism
\eqref{eqn:HamHJEq}, and it coincides with the equation \eqref{eqn:Example2_HJEquation} obtained in
the Lagrangian formulation, that is,
$$
\sum_{i=1}^{m}\derpar{W^i}{x^i}
+ \frac{1}{2} \, \tilde{g}_{ij}^{\alpha\beta} \derpar{W^i}{u^\alpha}\,\derpar{W^j}{u^\beta} - V(x^i,u^\alpha) = 0 \, .
$$

\subsection{Minimal surfaces in dimension three}
\label{sec:Example3}

Let us consider the following problem: we look for smooth maps $\phi \colon \R^2 \to \R$ such that
the set $\graph(\phi) \subseteq \R^3$, which is a surface in $\R^3$, has minimal area and satisfies
certain boundary conditions. This problem can be modeled as a first-order classical field theory
with base manifold $M = \R^2$ and $1$-dimensional fibers, that is, $E = \R^2 \times \R$. Let $(x,y)$
be the global coordinates in $M = \R^2$, and $(x,y,u)$ the global coordinates in $E = \R^3$. In
these coordinates, the canonical volume form $\eta \in \df^2(\R^2)$ is given by
$\eta = \d x \wedge \d y$. Then, in the natural coordinates $(x,y,u,u_1,u_2)$ of $J^1\pi$, the
Lagrangian density $\Lag \in \df^{2}(J^1\pi)$ for this field theory is
\begin{equation}\label{eqn:Example3_LagrangianDensity}
\Lag = \sqrt{1 + u_1^2 + u_2^2} \, \d x \wedge \d y \, .
\end{equation}
Observe that $\Lag$ is a regular Lagrangian density, since the Hessian matrix of the Lagrangian
function $L = \sqrt{1 + u_1^2 + u_2^2}$ associated with $\Lag$ and $\eta$ is
$$
\left( \derpars{L}{u_i}{u_j} \right) = \frac{1}{\sqrt{(1 + u_1^2 + u_2^2)^3}}
\begin{pmatrix}
u_2^2 + 1 & -u_1u_2 \\ -u_1u_2 & u_1^2 + 1
\end{pmatrix} \, ,
$$
which has determinant
$$
\det\left( \derpars{L}{u_i}{u_j} \right) = \frac{1}{(1 + u_1^2 + u_2^2)^2} \neq 0 \, .
$$

\subsubsection*{Lagrangian formalism}

The local expression of the Cartan $2$-form $\Theta_\Lag \in \df^{2}(J^1\pi)$ is
\begin{equation}\label{eqn:Example3_Cartan2Form}
\Theta_\Lag = \frac{u_1}{\sqrt{1 + u_1^2 + u_2^2}} \d u \wedge \d y - \frac{u_2}{\sqrt{1 + u_1^2 + u_2^2}} \d u \wedge \d x
+ \frac{1}{\sqrt{1 + u_1^2 + u_2^2}} \d x \wedge \d y \, ,
\end{equation}
from where we deduce the following coordinate expression for the Cartan $3$-form
$\Omega_\Lag \in \df^{3}(J^1\pi)$
\begin{align}
\Omega_\Lag &= - \frac{u_2^2 + 1}{\sqrt{(1 + u_1^2 + u_2^2)^3}} \d u_1 \wedge \d u \wedge \d y \nonumber
+ \frac{u_1u_2}{\sqrt{(1 + u_1^2 + u_2^2)^3}} \d u_2 \wedge \d u \wedge \d y \\
&\qquad{} - \frac{u_1u_2}{\sqrt{(1 + u_1^2 + u_2^2)^3}} \d u_1 \wedge \d u \wedge \d x \label{eqn:Example3_Cartan3Form}
+ \frac{u_1^2 + 1}{\sqrt{(1 + u_1^2 + u_2^2)^3}} \d u_2 \wedge \d u \wedge \d x \\
&\qquad{} + \frac{u_1}{\sqrt{(1 + u_1^2 + u_2^2)^3}} \d u_1 \wedge \d x \wedge \d y \nonumber
+ \frac{u_2}{\sqrt{(1 + u_1^2 + u_2^2)^3}} \d u_2 \wedge \d x \wedge \d y \, .
\end{align}
Thus, a multivector field $\X_\Lag \in \vf^{2}(J^1\pi)$ solution to the Lagrangian field equation
\eqref{eqn:LagFieldEqMVF} is given in coordinates by
$$
\X_\Lag = f X_1 \wedge X_2 = f \left( \derpar{}{x} + u_1 \derpar{}{u} + F_{1,1}\derpar{}{u_1} + F_{1,2}\derpar{}{u_2} \right) \wedge
\left( \derpar{}{y} + u_2 \derpar{}{u} + F_{2,1}\derpar{}{u_1} + F_{2,2}\derpar{}{u_2} \right) \, ,
$$
where $f \in \Cinfty(J^1\pi)$ is a non-vanishing function, and the functions $F_{i,j}$ satisfy the
Euler-Lagrange equations for multivector fields \eqref{eqn:LagFieldEqLocal}, which in this case
reduce to the following single equation
\begin{equation}\label{eqn:Example3_EulerLagrangeEqMVF}
(1 + u_2^2)F_{1,1} - u_1u_2(F_{1,2} + F_{2,1}) + (1 + u_1^2)F_{2,2} = 0 \, .
\end{equation}

In addition, since we need the multivector fields solution to the field equation to be integrable
in order to give a suitable Hamilton-Jacobi formulation, we must require equations
\eqref{eqn:Integrability&SemiHolonomicity} to hold. In our case, since $\X_\Lag = X_1 \wedge X_2$,
we must require $[X_1,X_2] = 0$. Computing, we have
\begin{align*}
[X_1,X_2] &= (F_{1,2} - F_{2,1}) \derpar{}{u} \\
&\quad + \left( \derpar{F_{2,1}}{x} + u_1\derpar{F_{2,1}}{u} + F_{1,1}\derpar{F_{2,1}}{u_1} + F_{1,2} \derpar{F_{2,1}}{u_2} \right. \\
&\qquad{} - \left. \derpar{F_{1,1}}{y} - u_2\derpar{F_{1,1}}{u} - F_{2,1}\derpar{F_{1,1}}{u_1} - F_{2,2}\derpar{F_{1,1}}{u_2} \right) \derpar{}{u_1} \\
&\quad + \left( \derpar{F_{2,2}}{x} + u_1\derpar{F_{2,2}}{u} + F_{1,1}\derpar{F_{2,2}}{u_1} + F_{1,2} \derpar{F_{2,2}}{u_2} \right. \\
&\qquad{} - \left. \derpar{F_{1,2}}{y} - u_2\derpar{F_{1,2}}{u} - F_{2,1}\derpar{F_{1,2}}{u_1} - F_{2,2}\derpar{F_{1,2}}{u_2} \right) \derpar{}{u_2}
\end{align*}
Then, requiring this last expression to vanish, we have that the multivector fields $\X_\Lag$ which
are solutions to the field equation are integrable if, and only if, the following $3$ equations hold
\begin{align}
F_{1,2} - F_{2,1} = 0 \, , \nonumber \\
\derpar{F_{2,1}}{x} + u_1\derpar{F_{2,1}}{u} + F_{1,1}\derpar{F_{2,1}}{u_1} + F_{1,2} \derpar{F_{2,1}}{u_2}
\label{eqn:Example3_IntegrabilityLagrangian}
- \derpar{F_{1,1}}{y} - u_2\derpar{F_{1,1}}{u} - F_{2,1}\derpar{F_{1,1}}{u_1} - F_{2,2}\derpar{F_{1,1}}{u_2} = 0 \, , \\
\derpar{F_{2,2}}{x} + u_1\derpar{F_{2,2}}{u} + F_{1,1}\derpar{F_{2,2}}{u_1} + F_{1,2} \derpar{F_{2,2}}{u_2} \nonumber
- \derpar{F_{1,2}}{y} - u_2\derpar{F_{1,2}}{u} - F_{2,1}\derpar{F_{1,2}}{u_1} - F_{2,2}\derpar{F_{1,2}}{u_2} = 0 \, .
\end{align}
A particular solution to equations \eqref{eqn:Example3_IntegrabilityLagrangian} is given by
$$
F_{1,1} = \derpar{u_1}{x} \quad ; \quad
F_{1,2} = F_{2,1} = \derpar{u_1}{y} = \derpar{u_2}{x} \quad ; \quad
F_{2,2} = \derpar{u_2}{y} \, .
$$
Moreover, one can easily check that these functions are also a solution to the Euler-Lagrange equation
\eqref{eqn:Example3_EulerLagrangeEqMVF}. Therefore, a particular holonomic multivector field solution
to the field equation is given in coordinates by
\begin{equation}\label{eqn:Example3_LagrangianMVFParticular}
\X_\Lag = f \left( \derpar{}{x} + u_1 \derpar{}{u} + \derpar{u_1}{x}\derpar{}{u_1} + \derpar{u_2}{x}\derpar{}{u_2} \right) \wedge
\left( \derpar{}{y} + u_2 \derpar{}{u} + \derpar{u_1}{y}\derpar{}{u_1} + \derpar{u_2}{y}\derpar{}{u_2} \right) \, .
\end{equation}
In the following we use the particular solution given by \eqref{eqn:Example3_LagrangianMVFParticular}
with $f = 1$ to state the Lagrangian Hamilton-Jacobi problem.

First we state the generalized version of the Hamilton-Jacobi problem. Let $\Psi \in \Gamma(\pi^1)$
be a jet field locally given by $\Psi(x,y,u) = (x,y,u,\psi_1(x,y,u),\psi_2(x,y,u))$, with
$\psi_1,\psi_2 \in \Cinfty(\R^3)$ being local functions. Then, from Proposition
\ref{prop:LagGenHJEquivalences} we know that $\Psi$ is a solution to the generalized
Lagrangian Hamilton-Jacobi problem if, and only if, the Euler-Lagrange multivector field $\X_\Lag$
given locally by \eqref{eqn:Example3_LagrangianMVFParticular} is tangent to the submanifold
$\Im(\Psi) \hookrightarrow J^1\pi$, which is locally defined by the constraints $u_j - \psi_j = 0$,
$j = 1,2$. Then, from the conditions $\Lie(X_i)(u_j - \psi_j) = 0$ we obtain the following systems
of $4$ partial differential equations
$$
\psi_1\derpar{\psi_1}{u} = 0 \quad ; \quad \psi_1\derpar{\psi_2}{u} = 0 \quad ; \quad
\psi_2\derpar{\psi_1}{u} = 0 \quad ; \quad \psi_2\derpar{\psi_2}{u} = 0 \, ,
$$
which admits the following set of solutions
\begin{equation}\label{eqn:Example3_GenLagHJSolutions}
\{ \psi_1(x,y,u) = f_1(x,y) , \psi_2(x,y,u) = f_2(x,y) \} \, ,
\end{equation}
where $f_1,f_2 \in \Cinfty(\R^3)$ are functions depending only on the coordinates of the base
manifold $M = \R^2$, that is, they are constant along the fibers of the bundle $\R^3 \to \R^2$.
Observe that both functions may be constant, and even vanish everywhere.

In order to obtain the equations of the Lagrangian Hamilton-Jacobi problem, we require in addition
the jet field $\Psi \in \Gamma(\pi^1)$ to satisfy the condition $\Psi^*\Omega_\Lag = 0$. Bearing in
mind the coordinate expression of the Cartan $3$-form $\Omega_\Lag \in \df^{3}(J^1\pi)$ given in
\eqref{eqn:Example3_Cartan3Form}, the pull-back $\Psi^*\Omega_\Lag \in \df^{3}(\R^3)$ by a jet field
$\Psi(x,y,u) = (x,y,u,\psi_1,\psi_2)$ gives
$$
\Psi^*\Omega_\Lag = \frac{1}{\sqrt{(1 + \psi_1^2 + \psi_2^2)^3}}
\left( (\psi_2^2 + 1)\derpar{\psi_1}{x} - \psi_1\psi_2 \left( \derpar{\psi_2}{x} + \derpar{\psi_1}{y} \right)
+ (\psi_1^2 + 1) \derpar{\psi_2}{y} \right) \d u \wedge \d x \wedge \d y \, .
$$
Therefore, $\Psi^*\Omega_\Lag = 0$ if, and only if, the following partial differential equation holds
$$
(\psi_2^2 + 1)\derpar{\psi_1}{x} - \psi_1\psi_2 \left( \derpar{\psi_2}{x} + \derpar{\psi_1}{y} \right)
+ (\psi_1^2 + 1) \derpar{\psi_2}{y} = 0 \, .
$$
It is easy to check that, from the functions in the set \eqref{eqn:Example3_GenLagHJSolutions} of
solutions to the generalized Lagrangian Hamilton-Jacobi problem, the following functions are solutions
to the Lagrangian Hamilton-Jacobi problem:
$$
\{ \psi_1 = \bar{f}(y) , \psi_2 = f(x) \} \,
$$
with $f,\bar{f} \in \Cinfty(E)$ being functions depending only on the coordinate $x$ and $y$, respectively
(with possibly one or both of them vanishing). In addition, when $\psi_i = g(x,y)$, we do not obtain
a closed formula for $\psi_j$ in terms of the function $g$, but there may be functions satisfying the
arising partial differential equation.

Finally, we state the Hamilton-Jacobi equation in the Lagrangian formalism. Since the jet field
$\Psi \in \Gamma(\pi^1)$ is a solution to the Lagrangian Hamilton-Jacobi problem, we have that
$\d(\Psi^*\Theta_\Lag) = 0$. Thus, there exists a $1$-form $\omega \in \df^{1}(E)$ given locally by
$\omega = W^1 \d y - W^2 \d x$ such that $\d\omega = \Psi^*\Theta_\Lag$. The pull-back of the Cartan
$2$-form $\Theta_\Lag$ by $\Psi$ gives in coordinates
$$
\Psi^*\Theta_\Lag = \frac{1}{\sqrt{1 + \psi_1^2 + \psi_2^2}}
\left( \psi_1 \d u \wedge \d y - \psi_2 \d u \wedge \d x + \d x \wedge \d y \right) \, .
$$
Hence, requiring $\Psi^*\Theta_\Lag = \d \omega$, we obtain
$$
\derpar{W^1}{x} + \derpar{W^2}{y} = \frac{1}{\sqrt{1 + \psi_1^2 + \psi_2^2}} \quad ; \quad
\derpar{W^1}{u} = \frac{\psi_1}{\sqrt{1 + \psi_1^2 + \psi_2^2}} \quad ; \quad
\derpar{W^2}{u} = \frac{\psi_2}{\sqrt{1 + \psi_1^2 + \psi_2^2}} \, ,
$$
which may be combined in the following single equation
\begin{equation}\label{eqn:Example3_HJEquation}
\derpar{W^1}{x} + \derpar{W^2}{y} = \sqrt{1 - \left(\derpar{W^1}{u}\right)^2 - \left(\derpar{W^2}{u}\right)^2} \, .
\end{equation}
This is the Hamilton-Jacobi equation for this field theory.

\subsubsection*{Hamiltonian formalism}

Now we state the Hamiltonian formulation of the Hamilton-Jacobi problem for this field theory.
In the natural coordinates $(x,y,u,p^1,p^2)$ of $J^1\pi^*$, the restricted Legendre map
$\Leg \colon J^1\pi \to J^1\pi^*$ associated to the Lagrangian density $\Lag$ given by
\eqref{eqn:Example3_LagrangianDensity} has the following coordinate expression
$$
\Leg^*p^1 = \frac{u_1}{\sqrt{1 + u_1^2 + u_2^2}} \quad ; \quad
\Leg^*p^2 = \frac{u_2}{\sqrt{1 + u_1^2 + u_2^2}} \, .
$$
From this last expression we deduce the coordinate expression of the extended Legendre map, which is
$$
\widetilde{\Leg}^*p^1 = \frac{u_1}{\sqrt{1 + u_1^2 + u_2^2}} \quad ; \quad
\widetilde{\Leg}^*p^2 = \frac{u_2}{\sqrt{1 + u_1^2 + u_2^2}} \quad ; \quad
\widetilde{\Leg}^*p = -\frac{1}{\sqrt{1 + u_1^2 + u_2^2}} \, ,
$$
as well as the coordinate expression of the (local) inverse map $\Leg^{-1} \colon J^1\pi^* \to J^1\pi$
$$
(\Leg^{-1})^*u_1 = \frac{p^1}{\sqrt{1 - (p^1)^2 - (p^2)^2}} \quad ; \quad
(\Leg^{-1})^*u_2 = \frac{p^2}{\sqrt{1 - (p^1)^2 - (p^2)^2}} \, .
$$
Observe that the inverse Legendre map is not defined on the points of $J^1\pi^*$ satisfying
$(p^1)^2 + (p^2)^2 = 1$.

The local Hamiltonian function associated to the Lagrangian formulation is then given by
$$
H(x,y,u,p^1,p^2) = - \sqrt{1 - (p^1)^2 - (p^2)^2} \, .
$$

Using the Hamiltonian section $h \in \Gamma(\mu)$ specified by this local Hamiltonian function
we define the Hamilton-Cartan forms $\Theta_h = h^*\Theta \in \df^{2}(J^1\pi^*)$,
$\Omega_h = h^*\Omega \in \df^{3}(J^1\pi^*)$, whose coordinate expressions are
\begin{align}\label{eqn:Example3_HamiltonCartan2Form}
&\Theta_h = p^1 \d u \wedge \d y - p^2 \d u \wedge \d x + \sqrt{1 - (p^1)^2 - (p^2)^2} \, \d x \wedge \d y \, , \\[10pt]
&\Omega_h = -\d p^1 \wedge \d u \wedge \d y + \d p^2 \wedge \d u \wedge \d x \nonumber \\
&\qquad{} + \displaystyle \frac{1}{\sqrt{1-(p^1)^2-(p^2)^2}} \label{eqn:Example3_HamiltonCartan3Form}
\left( p^1 \d p^1 \wedge \d x \wedge \d y + p^2 \d p^2 \wedge \d x \wedge \d y \right) \, .
\end{align}

Then, a locally decomposable $2$-vector field $\X_h \in \vf^{2}(J^1\pi^*)$ solution to the field
equation $\inn(\X_h)\Omega_h = 0$ is locally given by
\begin{equation*}
\begin{array}{l}
\X_h = f \left( \derpar{}{x} + \frac{p^1}{\sqrt{1-(p^1)^2-(p^2)^2}} \derpar{}{u} + G_1^1 \derpar{}{p^1} + G_1^2\derpar{}{p^2} \right) \\
\qquad{} \wedge \left( \derpar{}{y} + \frac{p^2}{\sqrt{1-(p^1)^2-(p^2)^2}} \derpar{}{u} + G_2^1 \derpar{}{p^1} + G_2^2\derpar{}{p^2} \right) \, ,
\end{array}
\end{equation*}
with the functions $G_i^j$ satisfying the Hamilton-De Donder-Weyl equations \eqref{eqn:HamFieldEqLocal},
which in this case reduce to the following single equation
\begin{equation}\label{eqn:Example3_HamiltonEqMVF}
G_1^1 + G_2^2 = 0 \, .
\end{equation}
Following the same procedure given in the Lagrangian formalism, an integrability condition must be
required to this multivector field. From \cite{art:Echeverria_Lopez_Marin_Munoz_Roman04} we know
that a particular choice of a locally decomposable and integrable multivector field solution to the
field equation is given in coordinates by
\begin{equation}\label{eqn:Example3_HamiltonianMVFParticular}
\begin{array}{l}
\displaystyle X_h = f \left( \derpar{}{x} + \frac{p^1}{\sqrt{1-(p^1)^2-(p^2)^2}} \derpar{}{u}
+ \derpar{p^1}{x} \derpar{}{p^1} + \derpar{p^2}{x}\derpar{}{p^2} \right) \\
\displaystyle \qquad{} \wedge \left( \derpar{}{y} + \frac{p^2}{\sqrt{1-(p^1)^2-(p^2)^2}} \derpar{}{u}
+ \derpar{p^1}{y} \derpar{}{p^1} + \derpar{p^2}{y}\derpar{}{p^2} \right) \, ,
\end{array}
\end{equation}
As in the Lagrangian formalism, in the following we use the particular solution given by
\eqref{eqn:Example3_HamiltonianMVFParticular} with $f = 1$ to state the Hamiltonian Hamilton-Jacobi
problem.

In order to state the generalized Hamiltonian Hamilton-Jacobi problem, let $s \in \Gamma(\pi_{\R^3}^r)$
be a section given in coordinates by $s(x,y,u) = (x,y,u,s^1(x,y,u),s^2(x,y,u))$, where
$s^1,s^2 \in \Cinfty(\R^3)$ are local functions. By Proposition \ref{prop:HamGenHJEquivalences},
the section $s$ is a solution to the generalized Hamiltonian Hamilton-Jacobi problem if, and only if,
the multivector field $\X_h$ given in coordinates by \eqref{eqn:Example3_HamiltonianMVFParticular}
is tangent to the submanifold $\Im(s) \hookrightarrow J^1\pi^*$ defined locally by the constraints
$p^j - s^j = 0$, $j = 1,2$. Then the tangency of $\X_h$ along the submanifold $\Im(s)$ gives the
following system of $4$ partial differential equations.
\begin{gather*}
\frac{s^1}{\sqrt{1-(s^1)^2-(s^2)^2}}\derpar{s^1}{u} = 0 \quad ; \quad
\frac{s^1}{\sqrt{1-(s^1)^2-(s^2)^2}}\derpar{s^2}{u} = 0 \, , \\
\frac{s^2}{\sqrt{1-(s^1)^2-(s^2)^2}}\derpar{s^1}{u} = 0 \quad ; \quad
\frac{s^2}{\sqrt{1-(s^1)^2-(s^2)^2}}\derpar{s^2}{u} = 0 \, .
\end{gather*}
This system of partial differential equations admits the following set of local solutions
\begin{equation}\label{eqn:Example3_GenHamHJSolutions}
\{ s^1(x,y,u) = f^1(x,y) , s^2(x,y,u) = f^2(x,y) \} \, ,
\end{equation}
where $f^1,f^2 \in \Cinfty(\R^3)$ are functions depending only on the coordinates $(x,y)$ of the
base manifold $M = \R^2$, and satisfying $(f^1)^2 + (f^2)^2 \neq 1$.

Now, to obtain the equation of the Hamiltonian Hamilton-Jacobi problem, we require in addition
that the section $s \in \Gamma(\pi_{\R^3}^r)$ satisfies the condition $s^*\Omega_h = 0$ or,
equivalently, we require the $2$-form $h \circ s \in \df^{2}(E)$ to be closed. In coordinates,
bearing in mind the coordinate expression \eqref{eqn:Example3_HamiltonCartan3Form} of the
Hamilton-Cartan $3$-form $\Omega_h \in \df^{3}(J^1\pi^*)$, we have
$$
s^*\Omega_h = \left( \derpar{s^1}{x} + \derpar{s^2}{y} \right) \d u \wedge \d x \wedge \d y \, ,
$$
from where the condition $s^*\Omega_h = 0$ is locally equivalent to the equation:
$$
\derpar{s^1}{x} + \derpar{s^2}{y} = 0 \, .
$$
This additional equation restricts the set of solutions \eqref{eqn:Example3_GenHamHJSolutions} to
the following:
$$
\{ s^1 = \bar{f}(y) , s^2 = f(x) \} \, .
$$
As in the Lagrangian Hamilton-Jacobi problem, $f,\bar{f} \in \Cinfty(E)$ are functions depending only
on the coordinate $x$ and $y$, respectively, and they may vanish. In addition, when $s^i = g(x,y)$, we
do not obtain a closed formula for $s^j$ in terms of the function $g$, but there may be functions
satisfying the arising partial differential equation. Recall that both $f,\bar{f}$ must satisfy
$f^2 + \bar{f}^2 \neq 1$.

Finally, we state the Hamilton-Jacobi equation in the Hamiltonian formalism. Since the section
$s \in \Gamma(\pi_{\R^3}^r)$ is a solution to the Hamiltonian Hamilton-Jacobi problem, the form
$h \circ s \in \df^{2}(E)$ is closed. Thus, there exists a $1$-form $\omega \in \df^{1}(E)$
given locally by $\omega = W^1 \d y - W^2 \d x$ such that $\d\omega = h \circ s$. In coordinates,
the form $h \circ s = s^*\Theta_h$ is given by
$$
h \circ s = s^1 \d u \wedge \d y - s^2 \d u \wedge \d x + \sqrt{1 - (s^1)^2 - (s^2)^2} \, \d x \wedge \d y \, .
$$
Thus, requiring $h \circ s = \d \omega$, we obtain
$$
\derpar{W^1}{x} + \derpar{W^2}{y} = \sqrt{1 - (s^1)^2 - (s^2)^2} \quad ; \quad
\derpar{W^1}{u} = s^1 \quad ; \quad \derpar{W^2}{u} = s^2 \, ,
$$
which may be combined to obtain equation \eqref{eqn:Example3_HJEquation}, that is,
$$
\derpar{W^1}{x} + \derpar{W^2}{y} = \sqrt{1 - \left(\derpar{W^1}{u}\right)^2 - \left(\derpar{W^2}{u}\right)^2} \, .
$$

\section{Conclusions and further research}
\label{sec:Conclusions}

Starting from the geometric Hamilton-Jacobi theory developed mainly in
\cite{art:Carinena_Gracia_Marmo_Martinez_Munoz_Roman06,HJteam-2015} for mechanical systems and using
the results given in \cite{LMM-09} as standpoint, we have stated a geometric framework for first-order
classical field theories described in the multisymplectic setting.

The theory has been developed for the Lagrangian and the Hamiltonian formalisms. In both cases, first
we have stated the so-called {\sl generalized Hamilton-Jacobi problem}, which is the most natural one
in this geometrical ambient, and hence we have defined from it the standard {\sl Hamilton-Jacobi
problem}. Particular solutions to these problems are defined and characterized in several equivalent
ways and, in particular, one of these characterizations for the standard case in the Hamiltonian
formalism, when written in natural coordinates, leads the classical Hamilton-Jacobi equation for
field theories. After that, the definition and geometric characterization of complete solutions
is also given and different features about them are discussed. Finally, the equivalence between the
Lagrangian and the Hamiltonian Hamilton-Jacobi problems is also proved.

It is important to point out that this generalization of the theory has been achieved using
distributions in the jet bundles and multimomentum bundles where the Lagrangian and Hamiltonian
formalisms of multisymplectic classical field theories are developed. These are integrable distributions
whose integral sections are the solutions to the Lagrangian and Hamiltonian field equations, and they
are represented, in general, by means of equivalence classes of multivector fields. This choice has
enabled us to give a construction of the Hamilton-Jacobi theory in a very natural way. Thus, our model
is different from that given by L. Vitagliano in \cite{Vi-10} for higher-order field theories, who
uses connections as the main geometrical tool and a unified formalism to describe the Lagrangian and
Hamiltonian formulations at once.

We have analyzed several examples. First, non-autonomous mechanical systems can be considered as a
special situation of field theories, and hence the Hamilton-Jacobi equation for these systems has been
recovered from our model as a particular case. Second, we have  applied our results to obtain the
Hamilton-Jacobi equation for field theories described by quadratic affine Lagrangians. Finally, we
have written this equation for a more particular example: minimal surfaces in dimension three.

As further research, we believe that our geometric framework for the Hamilton-Jacobi theory can be
extended to higher-order field theories using the formulations in
\cite{art:Campos_DeLeon_Martin_Vankerschaver09,art:Prieto_Roman14}, thus generalizing the results of
\cite{art:Colombo_DeLeon_Prieto_Roman14_JPA,CLPR2} for higher-order mechanics and giving a different
but equivalent perspective to that of \cite{Vi-10} for this kind of theories.

A very relevant application of the Hamilton-Jacobi theory for first-order field theories would be to
the Palatini approach of General Relativity and, once the extension to higher-order field theory is
made, also to the Einstein-Hilbert Lagrangian approach, as well as to other gravitational theories.
As a previous step, a suitable multisymplectic description of these gravitational models must be done
and, although there are some recent attempts to do it \cite{Vey2014}, more work in this way is
necessary and research in this way is in progress.

Another interesting question in the ambient of the Hamilton-Jacobi theory is the existence of
conserved quantities (or conservation laws), and the integrability of the system. It is known (see
\cite{art:Carinena_Gracia_Marmo_Martinez_Munoz_Roman06} for more details) that, in the case of
mechanics, for a dynamical system with $n$ degrees of freedom, the existence of complete solutions
to the generalized Hamilton-Jacobi problem is  associated with the local existence of families of $n$
functions which are constants of motion. From a geometrical perspective, a complete solution is a
foliation in the fiber bundle which represents the phase space of the system, which is transverse to
the fibers, and such that the dynamical vector field is tangent to the leaves of this foliation. Then,
these leaves are locally the level sets of the functions which are constants of motion. In addition,
a complete solution to the Hamilton-Jacobi problem corresponds to a Lagrangian foliation (with respect
to the symplectic structure, canonical or not, which the phase space is endowed with), and thus the
constants of motion are in involution and the system is completely integrable.

In the case of field theories, from a geometrical point of view, the situation is quite similar:
as it is defined in Sections \ref{sec:LagHJCompleteSol} and \ref{sec:HamHJCompleteSol}, complete
solutions to the generalized Hamilton-Jacobi problem endow the jet and multimomentum bundles with
foliations which are transverse to the fibers and such that their leaves contain the image of the
sections solution to the field equations. When we consider just the Hamilton-Jacobi problem, these
are $m$-Lagrangian foliations (with respect to the corresponding multisymplectic structures, in the
sense defined in \cite{art:Cantrijn_Ibort_DeLeon99}). Nevertheless, up to our knowledge, the notion
of ``integrability'' is not clearly stated in these cases. Furthermore, although the leaves of these
foliations can be also locally defined as level sets of families of functions, how to associate these
functions with conservation laws in field theories must be investigated and, even if the foliation is
$m$-Lagrangian, these functions are not said ``to be in involution'' because, although there are
several attempts to define unambiguously a Poisson bracket (for functions) in covariant field theories
(that is, in multisymplectic geometry), this problem is not solved in a completely satisfactory way.
The discussion on all these topics is also under research.

\subsection*{Acknowledgments}

We acknowledge the financial support of the \textsl{Ministerio de Ciencia e Innovaci\'{o}n} (Spain),
projects MTM2011-22585 and MTM2011-15725-E, \textsl{Ministerio de Econom\'{\i}a y Competividad} (Spain)
project MTM2013-42870-P, the European project IRSES-project ``Geomech-246981'', the ICMAT Severo Ochoa
project SEV-2011-0087, the AGAUR project 2009 SGR:1338, and the project E24/1 (Gobierno de Arag\'on).
P.D. Prieto-Mart\'{\i}nez wants to thank the UPC for a Ph.D grant.


\end{document}